%% file: main.tex
\documentclass[aps,prx,twocolumn,superscriptaddress,nofootinbib,floatfix]{revtex4-1}

\input{mainDef}

\makeatletter
\newsavebox{\@brx}
\newcommand{\llangle}[1][]{\savebox{\@brx}{\(\m@th{#1\langle}\)}%
  \mathopen{\copy\@brx\kern-0.5\wd\@brx\usebox{\@brx}}}
\newcommand{\rrangle}[1][]{\savebox{\@brx}{\(\m@th{#1\rangle}\)}%
  \mathclose{\copy\@brx\kern-0.5\wd\@brx\usebox{\@brx}}}
\makeatother

\usepackage{multirow}

\usepackage{orcidlink} 
\usepackage{adjustbox}

\usepackage{graphicx}
\usepackage{enumerate}
\usepackage{epstopdf}
\usepackage{amsmath,amsthm,amssymb,epsfig,amsfonts,mathrsfs}

\hypersetup{colorlinks=true}

\def\Tr{\mathop{\mathrm{Tr}}\nolimits}
\newcommand{\QR}{\overrightarrow{Q}_{\rm KD}}
\newcommand{\QL}{\overleftarrow{Q}_{\rm KD}}
\newcommand{\QLR}{\overleftrightarrow{Q}_{\rm KD}}

\begin{document}

\flushbottom

\title{Temporal Kirkwood–Dirac Quasiprobability Distribution and Unification of Temporal State Formalisms through Temporal Bloch Tomography}

\author{Zhian Jia\orcidlink{0000-0001-8588-173X}}
\email{giannjia@foxmail.com}
\affiliation{Institute of Quantum Physics, School of Physics, Central South University,
Changsha 418003, China}
\affiliation{Centre for Quantum Technologies, National University of Singapore, SG 117543, Singapore}
\affiliation{Department of Physics, National University of Singapore, SG 117543, Singapore}

\author{Kavan Modi\orcidlink{0000-0002-2054-9901}}
\email{kavan@quantumlah.org}
\affiliation{Science, Mathematics and Technology Cluster, Singapore University of Technology and Design,
8 Somapah Road, 487372 Singapore}

\author{Dagomir Kaszlikowski}
\email{phykd@nus.edu.sg}
\affiliation{Centre for Quantum Technologies, National University of Singapore, SG 117543, Singapore}
\affiliation{Department of Physics, National University of Singapore, SG 117543, Singapore}


\begin{abstract}
Temporal quantum states generalize the multipartite density operator formalism to the time domain, enabling a unified treatment of quantum systems with both timelike and spacelike correlations. Despite a growing body of temporal state formalisms, their precise operational relationships and conceptual distinctions remain unclear. In this work, we resolve this issue by extending the Kirkwood-Dirac (KD) quasiprobability distribution to arbitrary multi-time quantum processes and, more broadly, to general spatiotemporal settings. We define left, right, and doubled temporal KD quasiprobabilities, together with their real components, which we identify as temporal Margenau-Hill (MH) quasiprobabilities. All of these quantities are experimentally accessible through interferometric measurement schemes. By characterizing their nonclassical features, we show that the generalized KD framework provides a unified operational foundation for a wide class of temporal state approaches and can be directly implemented via temporal or spatiotemporal Bloch tomography.
\end{abstract}

\maketitle


\emph{Introduction.} ---
The nonclassical features of quantum systems are manifold and play a central role in enabling quantum technologies such as quantum computation, communication, and sensing.  
Fundamental notions such as quantum entanglement \cite{Horodecki2009}, Bell nonlocality \cite{Brunner2014RMP}, Kochen–Specker contextuality \cite{Budroni2022KScontext}, measurement incompatibility \cite{Guhne2023incompatible}, uncertainty relations \cite{Busch2014RMPmeasure,Coles2017uncertainty}, the negativity or complex-valuedness of quasiprobability distributions \cite{ferrie2011quasi,Wigner1932on,Kirkwood1933quantum,Dirac1945on,ArvidssonShukur2024KDreview}, and monogamy relations \cite{Ramanathan2012,Jia2016,Jia2018entropic,jia2017exclusivity}, among others, provide key conceptual frameworks for characterizing the intrinsically nonclassical behavior of quantum systems across diverse experimental settings.

A powerful framework for capturing such nonclassical phenomena is the quasiprobabilistic formulation of quantum mechanics, originating from Wigner's seminal work~\cite{Wigner1932on}. Within this context, a natural measure of nonclassicality emerges from the Kirkwood–Dirac (KD) distribution, which encodes certain quantumness absent in classical statistical theories through complex-valued quasiprobability distributions.
The KD distribution was first introduced by Kirkwood in 1933~\cite{Kirkwood1933quantum} and later by Dirac in 1945~\cite{Dirac1945on}. For a pure state $|\psi\rangle$, it is defined as
\begin{equation}\label{eq:standardKD}
    Q_{\rm KD}(a,b) = \langle a|\psi\rangle \langle \psi | b\rangle \langle b|a\rangle = \Tr(\rho_{\psi} \Pi_a \Pi_b),
\end{equation}
where $\rho_{\psi} = |\psi\rangle\langle\psi|$ is the pure state density matrix, and $\Pi_a = |a\rangle\langle a|$, $\Pi_b = |b\rangle\langle b|$ are projectors onto eigenstates of observables $A$ and $B$, respectively.
Comprehensive reviews of the KD distribution and its applications can be found in~\cite{ArvidssonShukur2024KDreview}.
The real part is the KD quasiprobability distribution is called the  Margenau-Hill (MH) quasiprobability distribution \cite{margenau1961correlation}, which also have many applications in quantum information theory.

\begin{figure}[b]
\includegraphics[width=4.5cm]{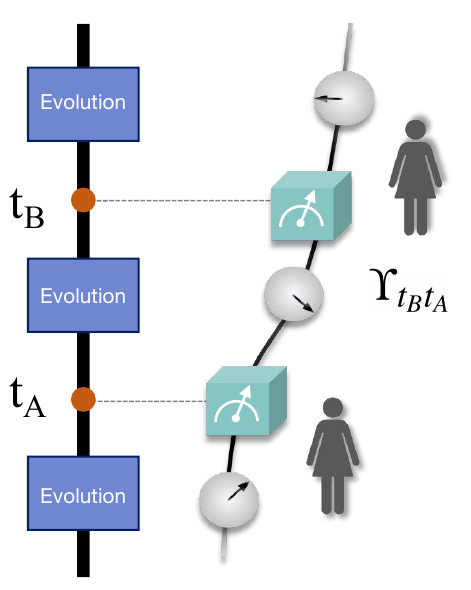}
\caption[Illustration of a temporal state at two time steps]{Illustration of a temporal state at two time steps, from $t_A$ to $t_B$. 
Pauli measurements are performed at these times to carry out temporal state tomography. 
In the PDO framework, one obtains the LvN distribution, from which the corresponding PDO can be reconstructed via tomography. 
In the right temporal KD case, the corresponding right temporal KD quasiprobability distribution is obtained, from which the joint expectation value 
$\langle \{ \sigma_{\mu}(t_B), \sigma_{\nu}(t_A) \} \rangle$ 
is computed. 
This procedure yields the temporal KD state $\overrightarrow{\Upsilon}_{t_B t_A}$. 
The left and doubled KD cases, as well as the left/right and doubled MH cases, proceed analogously.}
\label{fig:TemporalState}
\end{figure}

Temporal quantum phenomena have recently attracted significant attention, with applications in quantum causal modeling, quantum tomography, and quantum control, etc., being actively investigated. 
The use of the KD quasiprobability distribution in temporal settings, however, remains relatively unexplored. 
While it has found applications in various aspects of quantum dynamics 
(see~\cite{ArvidssonShukur2024KDreview,Gherardini2024KD,Lostaglio2023kirkwooddirac} for reviews), 
including out-of-time-ordered correlators (OTOCs)~\cite{larkin1969otoc,Halpern2018otoc,Alonso2019otocKD}, 
the Leggett-Garg test of macroscopic realism~\cite{Dressel2011weakLG}, 
and the consistent-histories interpretation of quantum mechanics~\cite[Sec.~8.2]{ArvidssonShukur2024KDreview}—
all closely tied to temporal processes—
a systematic study of its temporal generalization has yet to be undertaken.

On the other hand, to treat space and time on an equal footing in quantum theory, several concrete spatiotemporal formalisms have been proposed.
These include consistent histories~\cite{griffiths1984consistent}, pseudo-density operators (PDOs)~\cite{fitzsimons2015quantum}, quantum-classical games~\cite{gutoski2007toward}, quantum combs~\cite{Chiribella2009comb}, process tensors~\cite{PhysRevA.97.012127}, process matrices~\cite{oreshkov2012quantum}, multiple-time states~\cite{Aharonov2009multi}, Leifer--Spekkens causal states~\cite{Leifer2013toward}, superdensity operators~\cite{cotler2018superdensity}, symmetric bloom states~\cite{fullwood2022quantum,Parzygnat2023pdo}, and doubled density operators~\cite{jia2024spatiotemporal}, among others.  
The connections and distinctions among these formalisms constitute a crucial topic for further investigation~\cite{liu2023unification,Parzygnat2023pdo}.
Besides their fundamental significance, temporal states also play a key role in the study of information scrambling, quantum chaos, and related phenomena in quantum many-body physics.

In this paper, we systematically extend the KD quasiprobability distribution to temporal and spatiotemporal settings, focusing on multi-time processes of many qudits. 
We construct a discrete spacetime lattice by selecting specific points in a general multi-time quantum circuit. 
Measurements at these points yield spatiotemporal KD quasiprobability distributions, which are experimentally accessible. 
From these distributions, one can define a spatiotemporal state via (spatio)temporal Bloch tomography (Fig.~\ref{fig:TemporalState}); notably, this framework naturally unifies several existing temporal state formalisms.


\begin{table}
\caption{\label{tab:QuantumDistribution}Comparison of different temporal quantum distributions: Kirkwood–Dirac (KD), Margenau--Hill (MH) and L\"{u}ders–von Neumann (LvN).}
\centering
\resizebox{1.0\linewidth}{!}{%
\begin{tabular}{l|c|c|c}
\hline
\hline
 & KD  & MH  & LvN \\
\hline
Right & $\QR$ & $Q_{\rm MH}=(\QR+\QL)/2$ & $\times$ \\
\hline
Left & $\QL$ & $Q_{\rm MH}=(\QR+\QL)/2$ & $\times$  \\
\hline
Doubled & $\QLR$  & $\overleftrightarrow{Q}_{\rm MH}=(\QLR+\QLR^*)/2$ & diagonal of $\QLR$ \\
\hline
\hline
\end{tabular}
}
\end{table}


\emph{Temporal Kirkwood–Dirac quasiprobability distributions.}---
Consider a quantum process involving $(n+1)$ discrete time steps 
$\mathfrak{P} = (\rho_{t_0}, \mathcal{E}_{t_1 \leftarrow t_0}, \ldots, \mathcal{E}_{t_n \leftarrow t_{n-1}})$,
where $\rho_{t_0}$ is the initial state and each $\mathcal{E}_{t_j \leftarrow t_{j-1}}$ denotes a completely positive trace-preserving (CPTP) map describing the evolution from $t_{j-1}$ to $t_j$. 
The system state at an intermediate time $t_k$ is then given recursively by 
$\rho_{t_k} = \mathcal{E}_{t_k \leftarrow t_{k-1}} \circ \cdots \circ \mathcal{E}_{t_1 \leftarrow t_0}(\rho_{t_0})$.
The corresponding \emph{multi-time doubled temporal KD quasiprobability distribution} is defined as
\begin{widetext}
\begin{align}
   \overleftrightarrow{Q}_{\mathrm{KD}}(a_n, \ldots, a_0; b_n, \ldots, b_0)
   = \Tr \!\Big[
      \Pi_{a_n|A_n}^{t_n}\,
      \mathcal{E}_{t_n \leftarrow t_{n-1}}
      \Big(
         \cdots
         \Pi_{a_1|A_1}^{t_1}
         \big(
            \mathcal{E}_{t_1 \leftarrow t_0}
            ( \Pi_{a_0|A_0}^{t_0}\, \rho_{t_0}\, \Pi_{b_0|B_0}^{t_0} )
         \big)
         \Pi_{b_1|B_1}^{t_1}
         \cdots
      \Big)
      \Pi_{b_n|B_n}^{t_n}
   \Big].
\end{align}
\end{widetext}
Taking the left and right marginals yields the corresponding temporal KD quasiprobability distributions, $\QL(a_n,\ldots,a_0)$ and $\QR(b_n,\ldots,b_0)$, which are related by complex conjugation, $\QL^*(a_n,\ldots,a_0) = \QR(a_n,\ldots,a_0)$. 
Further details of their construction are provided in Supplemental Material. 
The L\"{u}ders--von Neumann (LvN) distribution emerges as a diagonal restriction of the doubled KD distribution,
\begin{equation}
    Q_{\mathrm{LvN}}(a_n,\ldots,a_0)
    = \overleftrightarrow{Q}_{\mathrm{KD}}(a_n,\ldots,a_0; a_n,\ldots,a_0).
\end{equation}
Taking the real part of the KD distributions yields the MH distributions. 
Since $\QL$ and $\QR$ are complex conjugates, their real parts coincide, resulting in a single (left/right) temporal MH distribution. 
A summary of the relationships among the various temporal quasiprobability distributions is presented in Table~\ref{tab:QuantumDistribution}.
The temporal KD distribution can be measured experimentally via interferometric scheme (see Supplemental Material) or quantum snapshotting protocol\cite{Wang2024KD,jia2026temporalstatetomographyquantum}.

A particularly important case is the two-time process 
$\mathfrak{P} = (\rho_{t_0}, \mathcal{E}_{t_1 \leftarrow t_0})$, 
which underlies a wide range of physical phenomena. 
Applications span diverse contexts, including the study of quantum chaos and OTOCs~\cite{larkin1969otoc,Halpern2018otoc,Alonso2019otocKD}, 
the analysis of two-time observables~\cite{breuer2002theory,gardiner2004quantum}, 
quantum thermodynamics~\cite{Santiago2024kirkwood}, 
and linear-response theory~\cite{Lostaglio2023kirkwooddirac}. 
The left and right two-time temporal KD quasiprobability distributions have been investigated in several recent works~\cite{Lostaglio2023kirkwooddirac,Santiago2024kirkwood}, 
where they are often regarded as conventional (spatial) KD distributions. 
Here, however, we stress that the temporal KD quasiprobability fundamentally encodes not only the statistical properties of the initial quantum state but also the dynamical features of the evolution map $\mathcal{E}_{t_1 \leftarrow t_0}$, also see \cite{Wang2024KD,Lie2025stateovertime,lie2025probingquantumstatesspacetime} for this perspective. 
This distinction becomes crucial when generalizing the framework to multi-time quantum processes.

The temporal KD quasiprobability distribution satisfies the Kolmogorov consistency condition. The marginalization over intermediate variables yields reduced KD distributions, as all quantum channels involved are CPTP maps. As a special case, each fixed-time distribution $p_{t_i}(a_i)$ corresponds to the physical measurement $\Pi_{a_i|A_i}^{t_i}$:
  $p_{t_i}(a_i) = \Tr \!\left( \rho_{t_i} \, \Pi_{a_i|A_i}^{t_i} \right)$,
and is thus a marginal of the temporal KD distribution.  
This ensures that the temporal KD distribution qualifies as a well-defined temporal quasiprobability distribution.  
Formally, we have the following lemma:

\begin{lemma}[Kolmogorov consistency condition]\label{lemma:marginal}
    Let $\mathfrak{P}_{t_0,\ldots, t_n} = (\rho_{t_0}, \mathcal{E}_{t_1 \leftarrow t_0}, \ldots, \mathcal{E}_{t_n \leftarrow t_{n-1}})$ be a quantum process over $n+1$ time steps. 
    Consider a sub-process $\mathfrak{P}'_{t_0,t_{i_1},\ldots,t_{i_k}}$ over $k+1$ time steps, where the initial state is the same and the evolution is a subset of that in $\mathfrak{P}_{t_0,\ldots,t_n}$, i.e., $\{t_{i_1},\ldots,t_{i_k}\} \subset \{t_1,\ldots,t_n\}$ with $t_{i_1} < \cdots < t_{i_k}$, and $\mathcal{E}_{t_{i_j} \leftarrow t_{i_{j-1}}}= \mathcal{E}_{t_{i_j} \leftarrow t_{i_{j}-1}} \circ \mathcal{E}_{t_{i_j} \leftarrow t_{i_{j}  -2  } } \circ  \cdots \circ \mathcal{E}_{t_{i_{j-1} +1}  \leftarrow t_{i_{j-1}}} $. 
    Then, the temporal KD quasiprobability distribution associated with $\mathfrak{P}'_{t_0,t_{i_1},\ldots,t_{i_k}}$ is obtained as the marginal of the temporal KD quasiprobability distribution associated with $\mathfrak{P}_{t_0,\ldots,t_n}$.
    If two subsets of time steps have a nonempty intersection, the corresponding marginals coincide on the overlapping region.
\end{lemma}


\emph{Quantumness of temporal Kirkwood–Dirac quasiprobability distribution.}---
Since temporal (and more generally spatiotemporal) KD quasiprobabilities can take negative, greater-than-one, or even complex values, these features are usually regarded as signatures of non-classicality, which we refer to as \emph{temporal KD non-classicality} or \emph{temporal KD quantumness}. In the spatial setting, KD non-classicality underlies quantum advantages in various quantum information tasks \cite{ArvidssonShukur2020KDmetrology,Halpern2018otoc,Alonso2019otocKD}.
Analogously, temporal KD non-classicality can be defined, and it may have potential applications in quantum information tasks involving temporal quantum processes.

As we have shown, the spatiotemporal KD quasiprobability distribution satisfies the Kolmogorov axioms, except that it may assume complex values outside the interval $[0,1]$.  
The negative or non-real values of $\overrightarrow{Q}_{\rm KD}(a_n, \ldots, a_0)$ are often referred to as ``non-classical'' or as indicating quantumness (the same applies to the other two types of generalized KD distributions). This non-classicality can be quantified by
\begin{equation}
    \mathcal{N}[\overrightarrow{Q}_{\rm KD}(a_n, \ldots, a_0)] = \sum_{a_0, \ldots, a_n} \left| \overrightarrow{Q}_{\rm KD}(a_n, \ldots, a_0) \right| - 1.
\end{equation}
The measure $\mathcal{N}$ depends on the initial state $\rho_{t_0}$, all quantum channels $\mathcal{E}_{t_j\leftarrow t_{j-1}}$, and the measurement settings.  
We emphasize that the quantumness of the spatiotemporal KD distribution depends on the underlying quantum evolutions, whereas the quantumness of the standard KD distribution is determined solely by the state and measurement settings.
This spatiotemporal KD quantumness thus captures the intrinsic spatiotemporal structure encoded within the distribution.

The temporal KD non-classicality measure has several important properties:

1. Non-negativity and faithfulness: 
    \[
        \mathcal{N}[\overrightarrow{Q}_{\rm KD}] \geq 0, 
     \]   
    for all $\QR$.    
   And
      $  \mathcal{N}[\overrightarrow{Q}_{\rm KD}] = 0 $
  if and only if  $\overrightarrow{Q}_{\rm KD}$  is classical.

2. Convexity over the initial state. For any $\lambda \in [0,1]$
    \begin{align}
        &\mathcal{N}\big[\QR[\lambda \rho_{t_0} +(1-\lambda) \omega_{t_0}]\big] \notag\\
        &\quad \leq \lambda \, \mathcal{N}[\QR[\rho_{t_0}]] + (1-\lambda) \, \mathcal{N}[\QR[\omega_{t_0}]].
    \end{align}
    
3. Convexity over quantum channels. For any $\lambda \in [0,1]$ and evolutions $\mathcal{E}_{t_j\leftarrow t_{j-1}}$ and $\mathcal{K}_{t_j\leftarrow t_{j-1}}$ from time step $t_{j-1}$ to $t_{j}$, we have
    \begin{align}
        &\mathcal{N}\Big[\QR[\lambda \mathcal{E}_{t_j\leftarrow t_{j-1}} 
        +(1-\lambda) \mathcal{K}_{t_j\leftarrow t_{j-1}}]\Big] \notag\\
        & \leq \lambda \, \mathcal{N}[\QR[\mathcal{E}_{t_j\leftarrow t_{j-1}}]] 
        + (1-\lambda) \, \mathcal{N}[\QR[\mathcal{K}_{t_j\leftarrow t_{j-1}}]].
    \end{align}

4. Decreasing under coarse-graining.
The coarse-graining of  $\QR(b_n,\cdots,b_0)$ is defined as $\QR^{\rm cg}(I_s,\cdots,I_0)=\sum_{\{b_k\in I_l\}} \QR(b_n,\cdots,b_0)$, where $\{I_l\}_{l=0}^s$ denotes a disjoint partition of $\{b_n,\cdots,b_0\}$. Then 
\begin{equation}
    \mathcal{N}[\QR^{\rm cg}(I_s,\cdots,I_0)]\leq \mathcal{N}[\QR(b_n,\cdots,b_0)].
\end{equation}

5. For the temporal KD quasiprobability $\QR(b_n,\cdots,b_0)$ and any of its marginals 
$\QR(b_{i_k},\cdots,b_{i_0})=\sum_{\{b_n,\cdots,b_0\}\setminus \{b_{i_k},\cdots,b_{i_0}\}}\QR(b_n,\cdots,b_0)$,
the resulting marginals exhibit a reduced degree of non-classicality. 
In other words, for an $(n+1)$-step multi-time quantum process, restricting to any $k$-step subset necessarily decreases the non-classicality of the temporal KD distribution.

6. For a product quantum process 
$\mathfrak{P}_1 \otimes \mathfrak{P}_2$ 
(in which both the initial states and the evolutions factorize),  
together with product measurement settings, the temporal KD negativity satisfies
\begin{align}
& \mathcal{N}\!\left(\QR(\mathfrak{P}_1 \otimes \mathfrak{P}_2)\right) 
 = \mathcal{N}\!\left(\QR(\mathfrak{P}_1)\right)
    \mathcal{N}\!\left(\QR(\mathfrak{P}_2)\right) \nonumber \\
 &\quad + \mathcal{N}\!\left(\QR(\mathfrak{P}_1)\right)
        + \mathcal{N}\!\left(\QR(\mathfrak{P}_2)\right)
        + 1 .
\end{align}
If we instead define the measure 
$\mathcal{N}' = \log \sum_{b_n,\ldots,b_0} |\QR(b_n,\ldots,b_0) |$,  
then the product rule becomes additive:
\begin{equation}
\mathcal{N}'\!\left(\QR(\mathfrak{P}_1 \otimes \mathfrak{P}_2)\right)
  = \mathcal{N}'\!\left(\QR(\mathfrak{P}_1)\right)
  + \mathcal{N}'\!\left(\QR(\mathfrak{P}_2)\right).
\end{equation}

All of the above statements also hold for the left temporal and double KD quasiprobability distributions, $\QL$ and $\QLR$.


A necessary and sufficient condition characterizing the quantumness of temporal KD quasiprobability distributions remains an open problem. 
Nevertheless, we establish the following partial result in this direction. 
For a given choice of measurement settings and dynamics, one may define temporal KD joint measurement operators (measurements in Heisenberg picture).  
For example, for the right temporal KD quasiprobability distribution we introduce
\begin{equation}\label{eq:MR}
\overrightarrow{M}_{b_n,\ldots,b_0} 
= \Pi_{b_0}^{t_0} \, 
  \mathcal{E}^{\dagger}_{t_1 \leftarrow t_0} \Bigl( 
    \Pi_{b_1}^{t_1} \, 
    \mathcal{E}^{\dagger}_{t_2 \leftarrow t_1} \bigl( 
        \Pi_{b_2}^{t_2} \cdots 
        \mathcal{E}^{\dagger}_{t_n \leftarrow t_{n-1}}(\Pi_{b_n}^{t_n}) 
    \bigr) 
  \Bigr),
\end{equation}
which admit a natural interpretation as back-evolving all measurement operators to the initial time $t_0$.  
The corresponding quasiprobability can then be written as
\begin{equation}
   \QR(b_n,\ldots,b_0)
   = \Tr\!\bigl( \overrightarrow{M}_{b_n,\ldots,b_0} \, \rho_{t_0} \bigr).
\end{equation}
Left and doubled constructions follow analogously.

\begin{theorem}[Classicality criterion]\label{thm:nogo}
Consider a multi-time quantum process
\(
\mathfrak{P}
= \bigl(\rho_{t_0}, \mathcal{E}_{t_1 \leftarrow t_0}, \ldots, 
        \mathcal{E}_{t_n \leftarrow t_{n-1}} \bigr),
\)
together with projective measurements 
$\{\Pi_{b_k}^{t_k}\}_{b_k}$ at each time step.
Suppose there exists a probability distribution 
$p(b_n,\ldots,b_0 \mid \mathfrak{P})$ satisfying the Kolmogorov axioms for all initial states $\rho_{t_0}$ such that  
(the temporal KD distribution with $\mathcal{N}[\QR]=0$ provides an example):
\begin{enumerate}
    \item convex linearity in $\rho_{t_0}$;
    \item correct marginals,
    \[
    \sum_{b_0,\,\ldots,\,\widehat{b_k},\,\ldots,\,b_n}
    p(b_n,\ldots,b_0 \mid \mathfrak{P})
    = \Tr\!\bigl(\rho_{t_k}\Pi_{b_k}^{t_k}\bigr), \quad \forall\, k,
    \]
    where $\rho_{t_k}$ denotes the output state at time $t_k$.
\end{enumerate}
Then the temporal joint measurement operators satisfy
\[
[M_{b_n,\ldots,b_1},M_{b_0}] = 0,
\quad
M_{b_n,\ldots,b_0}
= M_{b_n,\ldots,b_1} M_{b_0}
= M_{b_0} M_{b_n,\ldots,b_1}.
\]
If all channels are unitary, then each $M_{b_k}$ (marginals of  temporal joint measurement operators ) is itself a projector, and
\[
[M_{b_k},M_{b_l}] = 0,\ \forall\, k,l, 
\quad
M_{b_n,\ldots,b_0} = \prod_{j=0}^n M_{b_j}.
\]
Consequently, if $\mathcal{N}[\QR]>0$, then there must exist indices 
$b_0,\ldots,b_n$ for which the operators $M_{b_0},\ldots,M_{b_n}$ fail to commute.
\end{theorem}

The proof is given in Supplemental Material.  
For two-time processes, the result reduces to Theorem~1 of Ref.~\cite{Lostaglio2023kirkwooddirac}.  
Theorem~\ref{thm:nogo} thus clarifies the operational origin of nonclassicality in temporal KD quasiprobabilities.

\begin{table}
\caption{\label{tab:comp}Comparison of different quantum temporal state formalisms.}
\centering
\resizebox{1.0\linewidth}{!}{%
\begin{tabular}{l|c|c|c}
\hline
\hline
 & Hermitian & Local space   & Born rule \\
\hline
Left \& right KD temporal states & No & $\mathcal{H}_{(x,t)}$ & Left \& right KD distributions \\
\hline
Left/right MH temporal state & Yes & $\mathcal{H}_{(x,t)}$ & Left/right MH distribution \\
\hline
PDO & Yes & $\mathcal{H}_{(x,t)}$ & 2-time $=$ MH; $(n>2)$-time $\neq$ MH  \\
\hline
Doubled density operator & No  & $\mathcal{H}^L_{(x,t)}\otimes \mathcal{H}^R_{(x,t)}$ & Doubled KD distributions\\
\hline
Doubled MH temporal state & Yes  & $\mathcal{H}^L_{(x,t)}\otimes \mathcal{H}^R_{(x,t)}$ & Doubled MH distributions\\
\hline
\hline
\end{tabular}
}
\end{table}

\emph{Unification of temporal states through temporal Bloch tomography.}---
The temporal KD quasiprobability distribution offers a natural route to a unified framework for spatiotemporal tomography, placing several previously established spatiotemporal state formalisms under a single KD-based construction. Because the essential structure already appears in the purely temporal setting, we focus on temporal states below; the generalization to full spatiotemporal processes follows directly.

\begin{definition}[Temporal state]
For a multi-time quantum process $\mathfrak{P} = \bigl(\rho_{t_0}, \mathcal{E}_{t_1 \leftarrow t_0}, \dots, \mathcal{E}_{t_n \leftarrow t_{n-1}}\bigr)$, a \emph{temporal state} is an operator on the temporal Hilbert space $\bigotimes_{i=0}^n \mathcal{H}_{t_i}$ defined by
\begin{equation}
\Upsilon_{t_n, \cdots, t_0} = \mathcal{E}_{t_n \leftarrow t_{n-1}} \star_{\mathrm{TS}} (\cdots \star_{\mathrm{TS}} (\mathcal{E}_{t_1 \leftarrow t_0} \star_{\mathrm{TS}} \rho_{t_0})),
\end{equation}
where $\star_{\mathrm{TS}}$ denotes the temporal product operation, dependent on the chosen formalism. Temporal states satisfy the quantum Kolmogorov consistency condition: for any temporal subset $\{t_0, t_{i_1}, \dots, t_{i_k}\} \subseteq \{t_0, t_1, \dots, t_n\}$, the corresponding temporal state $\Upsilon_{t_{i_k}, \dots, t_{i_1},t_0}$ is the reduced state of $\Upsilon_{t_n, \cdots ,t_0}$.
\end{definition}

A canonical example is the PDO formalism \cite{fitzsimons2015quantum,jia2023quantumspace,liu2023unification,Liu2025PDO,Song2024PDO,Parzygnat2023pdo,Lie2025stateovertime,lie2025probingquantumstatesspacetime}, arising from \emph{temporal Bloch tomography}: one measures temporal joint expectation values of Pauli operators and reconstructs the object via Bloch expansion. Here, $\star_{\mathrm{TS}}$ is the Jordan product $A \star_{\mathrm{TS}} B = \{A,B\}/2$.
In a two-time scenario (Fig.~\ref{fig:TemporalState}), sequential Pauli measurements yield temporal LvN probabilities $p(a,b|\sigma_{\mu}^{t_1},\sigma_{\nu}^{t_0})$, from which
\begin{equation}
T^{\mu,\nu} = \big\langle \{\sigma_{\mu}^{t_1},\sigma_{\nu}^{t_0}\}\big\rangle = \sum_{a,b} ab\, p(a,b|\sigma_{\mu}^{t_1},\sigma_{\nu}^{t_0})
\end{equation}
follows. Inserting these correlators into the Bloch representation yields the PDO
\begin{equation}
R_{t_1,t_0}=\Upsilon^{\rm LvN}_{t_1,t_0} = 2^{-2} \sum_{\mu,\nu} T^{\mu,\nu}\,\sigma_{\mu}\otimes\sigma_{\nu},
\end{equation}
which is Hermitian but not generally positive semidefinite.

Replacing LvN probabilities with temporal KD quasiprobabilities yields a natural generalization.  
Right KD measurements give the right KD temporal correlators
\begin{equation}
\overrightarrow{T}^{\mu_n,\ldots,\mu_0}
= \sum_{a_n,\ldots,a_0} 
a_n \cdots a_0 \;
\overrightarrow{Q}_{\rm KD}\!\big(
a_n,\ldots,a_0 \,\big|\, 
\sigma_{\mu_n},\ldots,\sigma_{\mu_0}
\big),
\end{equation}
and analogously one obtains left KD correlators  
$\overleftarrow{T}^{\nu_n,\ldots,\nu_0}$  
and doubled correlators  
$\overleftrightarrow{T}^{\mu_n,\ldots,\mu_0;\,\nu_n,\ldots,\nu_0}$.  
Temporal states follow from the Bloch representation; for the doubled case,
\begin{equation}
\overleftrightarrow{\Upsilon}
= \frac{1}{d^{\,n+1}}
\sum_{\substack{\mu_0,\ldots,\mu_n\\\nu_0,\ldots,\nu_n = 0}}^{d^2-1}
\overleftrightarrow{T}^{\mu_n,\ldots,\mu_0;\,\nu_n,\ldots,\nu_0}
\biggl(\bigotimes_{i=0}^{n}\sigma_{\mu_i}\biggr)
\otimes
\biggl(\bigotimes_{j=0}^{n}\sigma_{\nu_j}\biggr).
\end{equation}
Left and right KD temporal states follow by replacing  
$\overleftrightarrow{T}$ with $\overleftarrow{T}$ or $\overrightarrow{T}$.  
Applying the same construction to the real parts of the KD distributions yields temporal MH states 
$\Upsilon^{\rm MH}$ and $\overleftrightarrow{\Upsilon}^{\rm MH}$.
Each of these states extends straightforwardly to general spatiotemporal processes.






For a local system of dimension $d$, we consider generalized Pauli (Hilbert--Schmidt) operators, satisfying:
(i) $\sigma_0 = \mathds{I}$; 
(ii) $\Tr(\sigma_j)=0$ for $j\ge 1$; 
(iii) orthogonality $\Tr(\sigma_\mu\sigma_\nu)= d\,\delta_{\mu\nu}$.
These operators form an orthogonal basis of $\Herm(\mathcal{H})$, $\mathcal{H}=\mathbb{C}^d$.
Any tomographically complete Hilbert--Schmidt basis suffices; for example, ``light-touch operators''~\cite{Liu2025PDO,fullwood2025operatorrep}.
Temporal states constructed using different orthonormal bases are related by operator-space basis transformation.

\begin{theorem}
For a given spatiotemporal quantum process, the temporal states  satisfy the following properties:

1. \emph{KD spatiotemporal state:} The doubled KD spatiotemporal state coincides with doubled density operator \cite{jia2024spatiotemporal} and we have:
\begin{align}
    & \overrightarrow{\Upsilon} = \Tr_{L} \overleftrightarrow{\Upsilon}, \quad 
      \overleftarrow{\Upsilon} = \Tr_{R} \overleftrightarrow{\Upsilon}, \quad
      \overrightarrow{\Upsilon} = \overleftarrow{\Upsilon}^{\dagger}.
\end{align}
The fixed-time density operator $\rho_{t_k}$ can be obtained from these spatiotemporal states via partial trace:
\begin{equation}
  \rho_{t_k} = \Tr_{t_n, \ldots, \widehat{t_k}, \ldots, t_0} \overrightarrow{\Upsilon} 
               = \Tr_{t_n, \ldots, \widehat{t_k}, \ldots, t_0} \overleftarrow{\Upsilon},
\end{equation}
where $\widehat{t_k}$ indicates that the partial trace is taken over all time steps except $t_k$. 
Since the left and right KD spatiotemporal states are the respective reduced states of the doubled KD state, the equal-time density operator can likewise be obtained from $\overleftrightarrow{\Upsilon}$. 
Moreover, they satisfy the quantum Kolmogorov consistency condition.

2. \emph{MH spatiotemporal states:} 
The MH states are Hermitianized versions of the KD states,
\begin{equation}
    \Upsilon^{\rm MH} = \frac{1}{2}\big(\overleftarrow{\Upsilon} + \overleftarrow{\Upsilon}^{\dagger}\big), 
    \quad 
    \overleftrightarrow{\Upsilon}^{\rm MH} = \frac{1}{2}\big(\overleftrightarrow{\Upsilon} + \overleftrightarrow{\Upsilon}^{\dagger}\big).
\end{equation}
From property~1, the fixed-time density operator $\rho_{t_k}$ can also be obtained from the MH spatiotemporal states, which similarly satisfy the quantum analogue of the Kolmogorov consistency condition.

3. \emph{Spatiotemporal Born rule:} 
For the KD states (temporal case shown for simplicity), take the inner product of measurements and state gives the corresponding distributions, e.g., for right KD temporal state
\begin{align}
\QR(b_n, \ldots, b_0) &= 
\Tr \Big[ 
    (\Pi_{b_n} \otimes \cdots \otimes \Pi_{b_0}) 
    \, \overrightarrow{\Upsilon} 
\Big]
\end{align}
Similar relations hold for the MH states.
The LvN distribution can also be obtained from $\overleftrightarrow{\Upsilon}$ via taking inner product with $(\Pi_{a_n} \otimes \cdots \otimes \Pi_{a_0}) \otimes (\Pi_{a_n} \otimes \cdots \otimes \Pi_{a_0})$\cite{jia2024spatiotemporal}.
\end{theorem}

\begin{figure}[t]
    \centering
    \includegraphics[width=1.0\linewidth]{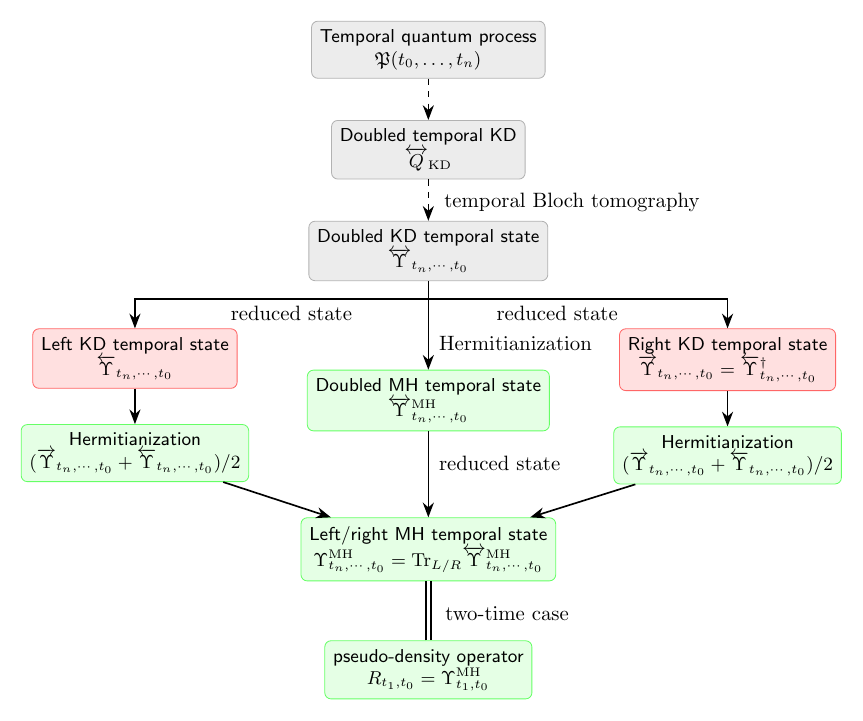}
\caption{The relationships between different temporal states arising from temporal Bloch tomography. Starting from the doubled temporal KD quasiprobability distributions, all temporal states can be obtained by taking partial trace or Hermitianization.}
    \label{fig:relation}
\end{figure}


For operators $M \in \mathbf{B}(\mathcal{H}_B \otimes \mathcal{H}_A)$ and 
$N \in \mathbf{B}(\mathcal{H}_C \otimes \mathcal{H}_B)$, we define (using the symbol $\star$ to emphasize its temporal composition)
$N \star M := (N_{CB} \otimes \mathds{I}_A)(\mathds{I}_C \otimes M_{BA})$,
where subscripts indicate the corresponding Hilbert spaces.
Given an initial density operator $\rho_{t_0} \in \mathbf{B}(\mathcal{H}_{t_0})$ and the Jamio\l kowski operator of a quantum channel,
$J[\mathcal{E}_{t_i \leftarrow t_{i-1}}] 
    = \sum_{k,l} 
    \mathcal{E}_{t_i \leftarrow t_{i-1}}(|k\rangle\!\langle l|)
    \otimes |l\rangle\!\langle k|$,
the temporal state associated with the KD quasiprobability distribution $\QR$ can be written as
\begin{equation} \label{eq:recurRep}
    \overrightarrow{\Upsilon}_{t_n, \cdots ,t_0}
    =
    J[\mathcal{E}_{t_n, \leftarrow, t_{n-1}}]
    \star \cdots \star
    J[\mathcal{E}_{t_1, \leftarrow, t_0}]
    \star \rho_{t_0}.
\end{equation}
Equivalently, the construction admits the recursive form
\begin{equation}
    \overrightarrow{\Upsilon}_{t_k, \cdots, t_0}
    =
    J[\mathcal{E}_{t_k, \leftarrow, t_{k-1}}]
    \star
    \overrightarrow{\Upsilon}_{t_{k-1}, \cdots, t_0},
    \qquad
    \overrightarrow{\Upsilon}_{t_0} = \rho_{t_0}.
\end{equation}
See supplementary material for a poof.
Using the MH quasiprobability distribution, the left/right MH temporal state is
\begin{equation}\label{eq:MHtemporalSate}
\Upsilon^{\rm MH} = \frac{1}{2} \big( \overrightarrow{\Upsilon} + \overleftarrow{\Upsilon} \big) 
                   = \frac{1}{2} \big( \overrightarrow{\Upsilon} + \overrightarrow{\Upsilon}^{\dagger} \big),
\end{equation}
with the corresponding temporal Born rule
\begin{equation}
Q_{\rm MH}(b_n, \ldots, b_0) = \Tr \big[ (\Pi_{b_n} \otimes \cdots \otimes \Pi_{b_0}) \, \Upsilon^{\rm MH} \big].
\end{equation}
This construction connects naturally to the PDO formalism~\cite{fitzsimons2015quantum}. In the two-time case, $\Upsilon^{\rm MH}$ coincides with $R_{t_1 t_0}$~\cite{zhao2018geometry,fullwood2022quantum,Parzygnat2023pdo,Liu2025PDO}; for general multi-time processes, the left/right MH temporal state differs from the PDO, see Supplemental Material for more details.

Connections and distinctions between the various spatiotemporal state formalisms are summarized in Table~\ref{tab:comp} and Figure~\ref{fig:relation}.

\vspace{1em}
\emph{Discussion.}---
In this work, we generalize the KD and MH quasiprobability distributions to the temporal and spatiotemporal settings and demonstrate that they play a crucial role in unifying different temporal (and spatiotemporal) state formalisms via temporal (and spatiotemporal) tomography.  
Our findings open several potential avenues for application, including providing deeper insights into temporal correlations within both the temporal KD framework and the temporal state formalism, as well as applications to Leggett–Garg tests and quantum metrology.

There are also several intriguing directions that warrant further exploration. 
In particular, it would be interesting to extend our results to the general process matrix formalism, 
from which one can further define a spatiotemporal KD quasiprobability distribution. 
The relation between the quantumness of the temporal KD distribution and temporal correlations in the temporal state formalism also deserves further investigation. 
Moreover, our definition of the temporal KD distribution can be naturally extended to the continuous-variable setting, 
which may have applications in the study of temporal quantum processes in quantum optical systems. 
All these topics are left for future work.

\begin{acknowledgments}
We thank Nicole Y. Halpern, Seok Hyung Lie, and Rathindra Das, for their valuable comments on the manuscript. Z. J. also thanks Zhi-Hao Ma and Yunlong Xiao for beneficial feedback.
Z. J. and D. K. are supported by the National Research Foundation, Singapore, and A*STAR under its CQT Bridging Grant and CQT–Return of PIs EOM YR1-10 Funding. Z. J. is also supported by the CQT Young Researcher Career Development Grant, start-up grant of Central South University (Grant No. 502045031) and the Frontier Interdisciplinary Direction ``Quantum Technologies'' Project of Central South University (Grant No. 506010805).
\end{acknowledgments}

\bibliographystyle{apsrev4-1-title}
\bibliography{Jiabib}

\end{document}


\flushbottom

\title{Supplemental Material:
Temporal Kirkwood–Dirac Quasiprobability Distribution and Unification of Temporal State Formalisms through Temporal Bloch Tomography}

\author{Zhian Jia\orcidlink{0000-0001-8588-173X}}
\email{giannjia@foxmail.com}
\affiliation{Institute of Quantum Physics, School of Physics, Central South University,
Changsha 418003, China}
\affiliation{Centre for Quantum Technologies, National University of Singapore, SG 117543, Singapore}
\affiliation{Department of Physics, National University of Singapore, SG 117543, Singapore}

\author{Kavan Modi\orcidlink{0000-0002-2054-9901}}
\email{kavan@quantumlah.org}
\affiliation{Science, Mathematics and Technology Cluster, Singapore University of Technology and Design,
8 Somapah Road, 487372 Singapore}

\author{Dagomir Kaszlikowski}
\email{phykd@nus.edu.sg}
\affiliation{Centre for Quantum Technologies, National University of Singapore, SG 117543, Singapore}
\affiliation{Department of Physics, National University of Singapore, SG 117543, Singapore}



\maketitle






In this supplementary material, we provide additional details on the spatiotemporal KD quasiprobability distribution, the interferometric scheme used to measure it, and several related technical aspects.

\section{Spatiotemporal generalization of Kirkwood–Dirac and Margenau-Hill quasiprobability distributions}
\label{sec:KDintro}

In this section, we present a detailed discussion on extending the Kirkwood–Dirac (KD) and Margenau–Hill (MH) quasiprobability distributions to the temporal and spatiotemporal domains. The resulting temporal quasiprobability distributions capture not only the properties of the initial state but also those of the underlying dynamics, in contrast to the spatial case.

\subsection{Two-time temporal Kirkwood–Dirac and Margenau-Hill quasiprobability distributions
}

Consider a two-time process $\mathfrak{P}=(\rho_{t_0},\mathcal{E}_{t_1\leftarrow t_0})$, where $\rho_{t_0}$ is the initial state and $\mathcal{E}_{t_1\leftarrow t_0}$ is a quantum channel describing the evolution. Such processes play a crucial role in many physical applications, the one which KD distribution can be applied includes quantum chaos and it OTOC characterization~\cite{larkin1969otoc,Halpern2018otoc,Alonso2019otocKD}, two-time observables~\cite{breuer2002theory,gardiner2004quantum}, quantum thermodynamics~\cite{Santiago2024kirkwood}, and linear response theory~\cite{Lostaglio2023kirkwooddirac}.

The KD quasiprobability distribution can be extended to the two-time setting by performing a measurement $\Pi_{b_0}^{t_0}$ at $t_0$ followed by a measurement $\Pi_{b_1}^{t_1}$ at $t_1$:
\begin{equation}
    \QR(b_1,b_0)= \Tr \!\left[(\mathcal{E}_{t_1\leftarrow t_0} (\rho_{t_0} \Pi_{b_0}^{t_0})) \, \Pi_{b_1}^{t_1}\right].
\end{equation}
This construction has appeared in several works, e.g.,~\cite{Lostaglio2023kirkwooddirac,Santiago2024kirkwood}, where it is typically interpreted as a standard KD quasiprobability distribution, namely as a generalized phase-space representation of a quantum state. In contrast, we stress that the temporal KD quasiprobability distribution intrinsically captures not only the information of the state but also the dynamics $\mathcal{E}_{t_1\leftarrow t_0}$, this viewpoint will become more crucial if we want to extend it to multi-time situation.

In this way, the standard KD quasiprobability distribution  may be interpreted in two distinct ways: (i) as an equal-time phase-space representation of a given state; or (ii) as a two-time temporal KD quasiprobability distribution with the evolution given by the identity channel. Both perspectives are of independent importance in different applications.

The temporal KD quasiprobability distribution introduced above is obtained by applying the measurement projectors on the bra space. Alternatively, one may define an analogous distribution by acting on the ket space:
\begin{equation}
    \QL(b_0,b_1)= \Tr \!\left[\Pi_{b_1}^{t_1}\, \mathcal{E}_{t_1\leftarrow t_0}(\Pi_{b_0}^{t_0} \rho_{t_0})\right],
\end{equation}
which is in fact the complex conjugate of $\QR(b_0,b_1)$, namely,
$\QR(b_0,b_1)^* = \QL(b_0,b_1)$. This follows from the identity $(\Tr A)^*=\Tr(A^{\dagger})$ together with the Kraus decomposition of the channel, $\mathcal{E}_{t_1\leftarrow t_0}(\bullet) =\sum_x K_x (\bullet) K_x^{\dagger}$. For convenience, we refer to $\QR$ and $\QL$ as the right and left temporal KD quasiprobability distributions, respectively.

Analogous to the standard KD quasiprobability distribution, one may define the \emph{temporal MH quasiprobability distribution}~\cite{margenau1961correlation} as the real part of $\QR$:  
\begin{equation}
   \overrightarrow{Q}^{\rm MH}(b_1,b_0) = \operatorname{Re}\,\QR(b_1,b_0).
\end{equation}
Equivalently,
\begin{equation} \label{eq:MHtwotime}
   Q_{\rm MH}(b_1,b_0) = \tfrac{1}{2}\,\big(\QR(b_0,b_1) + \QL(b_0,b_1)\big).
\end{equation}
All of these constructions will be useful for the tomography of temporal states.

We may also introduce a Heisenberg-picture formulation of the right temporal KD quasiprobability distribution by employing right two-time measurement~\cite{breuer2002theory,gardiner2004quantum,Lostaglio2023kirkwooddirac}:  
\begin{equation}
     \MR_{b_1,b_0}
    = \Pi_{b_0}^{t_0}\,\mathcal{E}^{\dagger}_{t_1\leftarrow t_0}(\Pi_{b_1}^{t_1}),
\end{equation}
where $\mathcal{E}^{\dagger}_{t_1\leftarrow t_0}$ denotes the adjoint channel with respect to the Hilbert–Schmidt inner product $\langle A,B\rangle:=\Tr (A^{\dagger}B)$. This can be interpreted as the back-evolution of the measurement at $t_1$ so that all observables are represented at the initial time $t_0$. The KD quasiprobability distribution is then given by the expectation value of this operator with respect to the initial state:  
\begin{equation}
    \QR(b_1,b_0) = \Tr \!\left( \rho_{t_0} \MR_{b_1,b_0}\right).
\end{equation}
Analogously, one may define  left two-time measurement
\begin{equation}
    \ML_{b_1,b_0} = \mathcal{E}^{\dagger}_{t_1\leftarrow t_0}(\Pi_{b_1}^{t_1})\, \Pi_{b_0}^{t_0},
\end{equation}
so that  
\begin{equation}
    \QL(b_1,b_0) = \Tr \!\left( \ML_{b_1,b_0} \rho_{t_0} \right).
\end{equation}
The MH two-time measurement is then given by  
\begin{equation}
    M^{\mathrm{MH}}_{b_1,b_0} = \tfrac{1}{2}\,\big(\MR_{b_1,b_0} + \ML_{b_1,b_0}\big).
\end{equation}
Recall that standard two-time observables take the form $\mathcal{E}^{\dagger}(B(t))\,A(0)$~\cite{breuer2002theory,gardiner2004quantum,Lostaglio2023kirkwooddirac}.  
By applying the spectral decomposition, the operators $\ML_{b,a}$ naturally arise.  
Alternatively, if we consider $A(0)\,\mathcal{E}^{\dagger}(B(t))$ as a two-time observable, the operators $\MR_{b,a}$ emerge in the corresponding decomposition.  
A symmetrized version,
$\frac{1}{2}\big[\mathcal{E}^{\dagger}(B(t))\,A(0) + A(0)\,\mathcal{E}^{\dagger}(B(t))\big]$,
can thus be regarded as the MH two-time observable.

A crucial property of the temporal KD quasiprobability distribution is that the
physical measurement probabilities at each time can be recovered as marginals:
\begin{align}
    p_{t_0}(b_0) &= \Tr \!\left(\Pi_{b_0}^{t_0}\,\rho_{t_0}\right) \nonumber \\
                 &= \sum_{b_1} \QR(b_1,b_0) \;=\; \sum_{b_1} \QL(b_1,b_0),
\end{align}
\begin{align}
    p_{t_1}(b_1) &= \Tr \!\left(\Pi_{b_1}^{t_1}\,\mathcal{E}_{t_1\leftarrow t_0} (\rho_{t_0})\right) \nonumber \\
                 &= \sum_{b_0} \QR(b_1,b_0) \;=\; \sum_{b_0} \QL(b_1,b_0).
\end{align}
This property underlies the terminology ``temporal,'' as the distribution consistently reproduces the correct measurement statistics at both times.  
This implies that the temporal MH quasiprobability distribution (Eq.~\ref{eq:MHtwotime}) also reproduces the correct measurement statistics at both time instances.

The most general and information complete form of temporal KD quasiprobability distribution for two-time setting should have measurement acting on both ket and bra space.
The corresponding doubled temporal KD distribution is defined as  
\begin{equation}
\begin{aligned}
       & \QLR(a_1,a_0;b_1,b_0) \\
    = & \Tr \!\left( \Pi_{a_1}^{t_1}\,\mathcal{E}_{t_1\leftarrow t_0} \!\left(\Pi_{a_0}^{t_0}\,\rho_{t_0}\,\Pi_{b_0}^{t_0}\right)\,\Pi_{b_1}^{t_1} \right).
\end{aligned}
\end{equation}
This distribution contains the full information of the two-time quantum process.
By taking left and right marginals, one recovers the distributions $\QR(b_1,b_0)$ and $\QL(a_1,a_0)$.  
Tracing over the variables at $t_1$ yields the standard KD quasiprobability distribution $Q_{\rm KD}(a_0,b_0)$ for $\rho_{t_0}$; similarly, tracing over the variables at $t_0$ yields the standard KD distribution $Q_{\rm KD}(a_1,b_1)$ for $\rho_{t_1}$.

We may likewise introduce the doubled two-time measurement  
\begin{equation}
     \MLR_{a_1,a_0;b_1,b_0}
    = \Pi_{b_0}^{t_0}\,\mathcal{E}^{\dagger}_{t_1\leftarrow t_0}\!\left(\Pi_{b_1}^{t_1}\Pi_{a_1}^{t_1}\right)\,\Pi_{a_0}^{t_0}.
\end{equation}
 The distribution can then be expressed as  
\begin{equation}
    \QLR(a_1,a_0;b_1,b_0) = \Tr \!\left( \MLR_{a_1,a_0;b_1,b_0}\,\rho_{t_0}\right).
\end{equation}
Taking marginals (summing over the corresponding indices), one obtains  
\begin{equation}
   \MR_{b_1,b_0}=\sum_{a_1,a_0}\MLR_{a_1,a_0;b_1,b_0}, \,\,
   \ML_{a_1,a_0}=\sum_{b_1,b_0}\MLR_{a_1,a_0;b_1,b_0}.
\end{equation}
In this sense, the doubled temporal KD quasiprobability distribution distribution and the corresponding doubled temporal measurement are information complete. Another reason why the doubled temporal KD quasiprobability distribution can be regarded as information complete is that, for each time step, its corresponding marginals remain phase-space KD quasiprobability distributions, which contain the complete information of the state at that instant. In contrast, for the left or right temporal KD distributions, the marginals at each time step correspond only to measurements of the state using a fixed complete set of projective measurements. (One could, in principle, introduce information-complete measurement settings at a given step—for example, by employing two distinct sets of projective measurements or an information complete POVM—but to keep the discussion simple, we shall not consider such generalizations in this work.) These distributions are not information complete, as the information about quantum superpositions is erased.

The two-time doubled temporal MH quasiprobability distribution $\overleftrightarrow{Q}_{\rm MH}$ is defined as the real part of the doubled temporal KD quasiprobability distribution,  
$
\overleftrightarrow{Q}_{\rm MH} = \frac{1}{2}\left(\QLR + \QLR^{*}\right).
$
The left and right marginals yield the corresponding left and right temporal MH quasiprobability distributions.  
The two-time doubled temporal MH measurement is given by
\begin{equation}
    \MLR_{a_1,a_0;b_1,b_0}^{\rm MH}
    = \frac{1}{2}\left(
        \MLR_{a_1,a_0;b_1,b_0}
        + \MLR_{a_1,a_0;b_1,b_0}^{\dagger}
      \right).
\end{equation}
Similarly, the left and right marginals give the corresponding left and right temporal MH measurements.

The two-time L\"{u}ders--von Neumann (LvN) probability distribution~\cite{vonNeumann2018mathematical,luders1950measurementRule} is defined as
\begin{equation}
\begin{aligned}
     Q_{\rm LvN}(b_1,b_0) 
     &= \Tr \!\left(\Pi_{b_1}^{t_1}\, \mathcal{E}_{t_1\leftarrow t_0} (\Pi_{b_0}^{t_0} \rho \Pi_{b_0}^{t_0})\, \Pi_{b_1}^{t_1} \right)\\
     &= \Tr \!\left(\mathcal{E}_{t_1\leftarrow t_0} (\Pi_{b_0}^{t_0} \rho \Pi_{b_0}^{t_0})\, \Pi_{b_1}^{t_1} \right).
\end{aligned}
\end{equation}
It is clear that the past distribution can be recovered by taking the marginal:
\begin{equation}
    p_{t_0}(b_0) = \sum_{b_1} Q_{\rm LvN}(b_1,b_0),
\end{equation}
whereas the future distribution cannot, in general, be obtained via marginalization. That is,
\begin{equation}
     p_{t_1}(b_1) \neq \sum_{b_0} Q_{\rm LvN}(b_1,b_0).
\end{equation}
Equality holds in some special cases, e.g., when the measurement at $t_0$ commutes with the initial state.
The doubled temporal KD quasiprobability distribution has the property that the LvN distribution can be regarded as its diagonal:
\begin{equation}
    Q_{\rm LvN}(b_1,b_0) = \QLR(b_1,b_0; b_1,b_0),
\end{equation}
that is, by choosing identical measurement settings for the bra and ket spaces. 

\subsection{Multi-time temporal Kirkwood–Dirac and Margenau-Hill quasiprobability distributions}
\label{sec:subMultipleKD}

Having established the framework for the two-time case, we now proceed to the general setting of the multi-time scenario.
Consider a quantum process over $(n+1)$ time steps, specified by $\mathfrak{P}=(\rho_{t_0}, \mathcal{E}_{t_1 \leftarrow t_0}, \ldots, \mathcal{E}_{t_n \leftarrow t_{n-1}})$, where $\rho_{t_0}$ is the initial state and each $\mathcal{E}_{t_j \leftarrow t_{j-1}}$ is a completely positive trace-preserving (CPTP) map evolving the system from $t_{j-1}$ to $t_j$. For each time step $t_k$, we define state $\rho_{t_k} = \mathcal{E}_{t_k \leftarrow t_{k-1}} \circ \cdots \circ \mathcal{E}_{t_1 \leftarrow t_0} (\rho_{t_0})$.
At each time step $t_i$ ($i = 0, \ldots, n$), a projective von Neumann measurement $\{\Pi_{a_i|A_i}^{t_i}\}$ is performed, with observable $A_i$ having the spectral decomposition $A_i^{t_i} = \sum_{a_i} a_i\, \Pi_{a_i|A_i}^{t_i}$. Here $a_i$ denotes the possible outcomes of $A_i$, and the notation $a_i|A_i$ explicitly indicates this correspondence.

We first introduce the following definition of a temporal quasiprobability distribution, which is applicable to any generalized probabilistic theory.

\begin{definition}[Temporal quasiprobability distribution]
Consider a quantum process over $(n+1)$ time steps, specified by 
$\mathfrak{P} = (\rho_{t_0}, \mathcal{E}_{t_1 \leftarrow t_0}, \ldots, \mathcal{E}_{t_n \leftarrow t_{n-1}})$, 
where $\rho_{t_k}$ denotes the quantum state at time $t_k$.  
Each fixed-time state admits a corresponding phase-space representation $Q_{t_k}(x_k|\rho_{t_k})$, 
with $x_k$ labeling the phase-space parameters associated with $\rho_{t_k}$.  
A {temporal quasiprobability distribution} $Q_{t_n \cdots t_0}(x_n, \ldots, x_0)$ 
is defined as a quasiprobability distribution whose marginals reproduce the correct phase-space representations $Q_{t_k}(x_k|\rho_{t_k})$ at each time step and satisfy the Kolmogorov consistency conditions, namely:  
(i) For any subset $\mathcal{T}$ of time step set $\{t_n,\cdots,t_0\}$, the corresponding marginal distribution $Q_{\mathcal{T}}$ obtained by taking marginal of $Q_{t_n \cdots t_0}(x_n, \ldots, x_0)$  is itself a temporal quasiprobability distribution for those times.
(ii) For any two subsets $\mathcal{T},\mathcal{S}$ of time steps with overlapping times $\mathcal{T}\cap \mathcal{S}\neq \emptyset$, the corresponding marginals on the overlap coincide $\Tr_{\mathcal{T}\setminus \mathcal{T}\cap \mathcal{S} } Q_{\mathcal{T}}=\Tr_{\mathcal{S}\setminus \mathcal{T}\cap \mathcal{S} } Q_{\mathcal{S}}$, here by tracing we mean taking marginals.
\end{definition}

Notice that both the temporal KD and MH quasiprobability distributions satisfy the conditions in the above definition. 
In contrast, if we fix the distribution for each time step as physical measurement statistics $p_{t_k}=\Tr(\rho_{t_k}\Pi_{b_k})$,
the LvN probability distribution does not satisfies the above definition, as it generally exhibits one-sided signaling: tracing out the future reproduces the past distribution, whereas tracing out the past does not necessarily yield the correct future distribution due to measurement disturbance and state collapse. 
We refer to such temporal distributions as \emph{causality-sensitive (CS) temporal quasiprobability distributions}.

The two-time temporal KD quasiprobability distributions can be generalized to multi-time setting straightforwardly:
\begin{align}
   & \overrightarrow{Q}_{\rm KD}(a_n, \ldots, a_0) = \Tr\Big[ 
    \mathcal{E}_{t_n \leftarrow t_{n-1}} \Big( 
    \cdots 
    \mathcal{E}_{t_2 \leftarrow t_1} \big(
    \mathcal{E}_{t_1 \leftarrow t_0} ( 
    \rho_{t_0} \Pi_{a_0|A_0} ) 
    \Pi_{a_1|A_1} \big)  \Pi_{a_{2}|A_{2}}  \cdots 
    \Big) 
    \Pi_{a_n|A_n}
    \Big],\\
  &  \overleftarrow{Q}_{\rm KD}(a_n, \ldots, a_0) = \Tr\Big[  \Pi_{a_n|A_n}^{t_n}
    \mathcal{E}_{t_n \leftarrow t_{n-1}} \Big( 
    \cdots 
    \Pi_{a_{2}|A_{2}}^{t_2}  \mathcal{E}_{t_2 \leftarrow t_1}   \big(
   \Pi_{a_1|A_1}^{t_1}  \mathcal{E}_{t_1 \leftarrow t_0} ( 
   \Pi_{a_0|A_0}^{t_0}  \rho_{t_0} ) 
    \big)  \cdots 
    \Big) 
    \Big],\\
   & \overleftrightarrow{Q}_{\rm KD}(a_n, \ldots, a_0;b_n, \ldots, b_0) = \Tr\Big[  \Pi_{a_n|A_n}^{t_n}
    \mathcal{E}_{t_n \leftarrow t_{n-1}} \Big( 
    \cdots 
     \Pi_{a_1|A_1}^{t_1} \big(
    \mathcal{E}_{t_1 \leftarrow t_0} ( 
   \Pi_{a_0|A_0}^{t_0}  \rho_{t_0}  \Pi_{b_0|B_0}^{t_0}  ) 
    \big) \Pi_{b_1|B_1}^{t_1}  \cdots 
    \Big) \Pi_{b_n|B_n}^{t_n}
    \Big].
\end{align}
These are referred to as the \emph{right, left, and  doubled temporal KD quasiprobability distributions}, respectively. The temporal MH quasiprobability distribution is defined as real parts of these distributions.

The multi-time temporal KD and MH quasiprobability distributions defined above are closely related, as summarized in the following lemma:

\begin{lemma}\label{lemma:QDleftright}
The right, left, and doubled temporal KD and MH quasiprobability distributions satisfy the following relations:

(i) The left and right temporal KD distributions are related by complex conjugation:  
    \begin{equation}
         \overrightarrow{Q}_{\rm KD}(a_n, \ldots, a_0) =  \overleftarrow{Q}_{\rm KD}(a_n, \ldots, a_0)^*.
    \end{equation}
This implies that $Q_{\mathrm{MH}} = \tfrac{1}{2}(\QL + \QR)$, indicating that there is no distinction between the left and right MH quasiprobability distributions (we thus will not use arrow notation for this case).

(ii)The left (right) temporal KD quasiprobability distribution is the left (right) marginal of the doubled temporal KD quasiprobability distribution: 
    \begin{align}
        \overrightarrow{Q}_{\rm KD}(b_n, \ldots, b_0)=\sum_{a_n,\cdots,a_0}  \overleftrightarrow{Q}_{\rm KD}(a_n, \ldots, a_0;b_n, \ldots, b_0) \\
        \overleftarrow{Q}_{\rm KD}(a_n, \ldots, a_0)=\sum_{b_n,\cdots,b_0}  \overleftrightarrow{Q}_{\rm KD}(a_n, \ldots, a_0;b_n, \ldots, b_0).
    \end{align}
This means doubled temporal KD quasiprobability distribution is information complete.

(iii) For doubled temporal KD quasiprobability distribution $\QLR$, we have $\overleftrightarrow{Q}_{\rm KD}(a_n, \ldots, a_0;b_n, \ldots, b_0)^*=\overleftrightarrow{Q}_{\rm KD}(b_n, \ldots, b_0;a_n, \ldots, a_0)$, thus the doubled temporal MH quasiprobability distribution can be written as
\begin{equation}
\begin{aligned}
       & \overleftrightarrow{Q}_{\rm MH}(a_n, \ldots, a_0;b_n, \ldots, b_0) \\
        = &\frac{1}{2} [\overleftrightarrow{Q}_{\rm KD}(a_n, \ldots, a_0;b_n, \ldots, b_0)\\
        &\,\, \,\, + \overleftrightarrow{Q}_{\rm KD}(b_n, \ldots, b_0;a_n, \ldots, a_0)].
\end{aligned}
\end{equation}
\end{lemma}

\begin{proof}
  Relation (i) and (iii) follows from $(\Tr M)^* = \Tr M^{\dagger}$ together with the Kraus representation of the channels, whereas relation (ii) is immediate from the definitions.  
\end{proof}

The doubled temporal KD quasiprobability distribution contains the full information of the process, thus the L\"{u}ders–von Neumann distribution can also be recovered. One can apply the same projection on both sides of the doubled temporal KD distribution:
\begin{equation}
    Q_{\rm LvN}(a_n,\ldots,a_0) = \overleftrightarrow{Q}_{\rm KD}(a_n,\ldots,a_0; a_n,\ldots,a_0).
\end{equation}

Note that if all quantum channels are taken to be identity (i.e., do-nothing) channels, the temporal Kirkwood–Dirac (KD) quasiprobability distribution reduces to the standard KD distribution. Hence, every spatial KD quasiprobability distribution can be regarded as a special case within the temporal framework; however, the converse does not hold, or at least not in a natural way.


We can also introduce the right temporal KD joint measurement operator
\begin{equation}\label{eq:MR}
\overrightarrow{M}_{b_n,\ldots,b_0} 
= \Pi_{b_0}^{t_0} \, \mathcal{E}^{\dagger}_{t_1 \leftarrow t_0} \Bigl( 
    \Pi_{b_1}^{t_1} \, \mathcal{E}^{\dagger}_{t_2 \leftarrow t_1} \bigl( 
        \Pi_{b_2}^{t_2} \cdots 
        \mathcal{E}^{\dagger}_{t_n \leftarrow t_{n-1}}(\Pi_{b_n}^{t_n}) 
    \bigr) 
\Bigr),
\end{equation}
which can be interpreted as back-evolving all measurements to the initial time step $t_0$.  
Accordingly, the joint probability can be expressed as
\begin{equation}
    \QR(b_n,\ldots,b_0) = \Tr\!\bigl(\overrightarrow{M}_{b_n,\ldots,b_0} \, \rho_{t_0}\bigr).
\end{equation}
Taking marginals of $\overrightarrow{M}_{b_n,\ldots,b_0}$ is a well-defined operation, from which one can obtain
\[
\overrightarrow{M}_{b_{k_1},\ldots,b_{k_s}}, \qquad 
\{b_{k_1},\ldots,b_{k_s}\} \subset \{b_0,\ldots,b_n\}.
\]
The most special single-instance operator can also be obtained as marginals,
\[ \overrightarrow{M}_{b_k} = 
\mathcal{E}^{\dagger}_{t_1 \leftarrow t_0} \Bigl( 
    \mathcal{E}^{\dagger}_{t_2 \leftarrow t_1} \bigl( 
        \cdots 
        \mathcal{E}^{\dagger}_{t_k \leftarrow t_{k-1}}(\Pi_{b_k}^{t_k}) \bigr) \Bigr)  .
\]
Similarly, we can define the left multi-time temporal KD measurement operators $\ML_{b_n, \ldots, b_0}$ such that 
$\QL(b_n, \ldots, b_0) = \Tr\!\left(\ML_{b_n, \ldots, b_0}\,\rho_{t_0}\right)$.  
Likewise, we introduce the doubled multi-time temporal KD measurement operators 
$\overleftrightarrow{M}_{a_n, \ldots, a_0;\, b_n, \ldots, b_0}$ satisfying 
$\QLR(a_n, \ldots, a_0; b_n, \ldots, b_0) = 
\Tr\!\left(\overleftrightarrow{M}_{a_n, \ldots, a_0;\, b_n, \ldots, b_0}\,\rho_{t_0}\right)$.  
For the MH case, the corresponding multi-time temporal MH measurement operators  is obtained simply by symmetrization.

\subsection{Spatiotemporal Kirkwood–Dirac and Margenau-Hill quasiprobability distributions}

As shown in Figure~\ref{fig:QC}, for a multipartite quantum system with multiple time steps, the above discussion extends naturally to the spatiotemporal setting when local measurements are performed on each subsystem. We emphasize the locality of the measurement: if joint projectors are used instead, the resulting KD distribution reduces to a temporal KD distribution, since the whole system is then treated as a single particle.

\begin{figure}[b]
    \centering
    \includegraphics[width=0.4\linewidth]{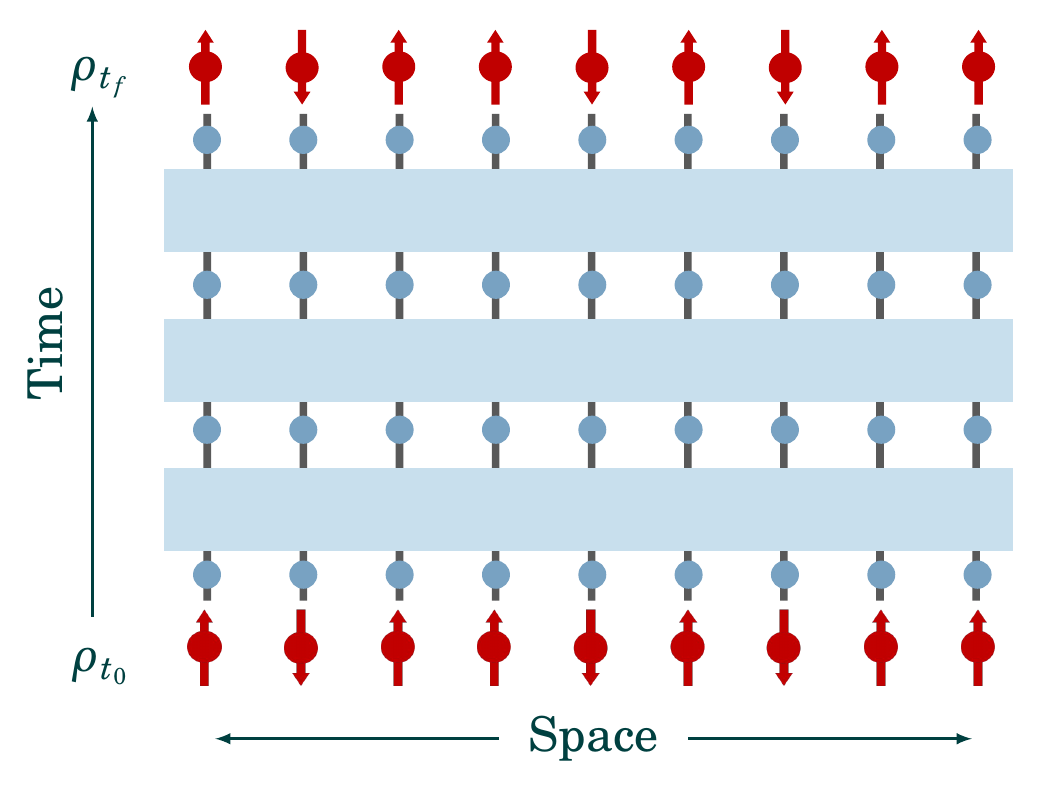}
\caption{Illustration of a spatiotemporal quantum process, where a multipartite initial state $\rho_{t_0}$ evolves over several time steps. The gray dots represent spacetime points where quantum operations, such as measurements, can be implemented. Since the initial multipartite state may be entangled, tracing out some local degrees of freedom results in a process that goes beyond the simple picture of a mixed state evolving under a quantum channel.}
    \label{fig:QC}
\end{figure}

As a concrete example, consider the two-time bipartite case with initial state $\rho_{t_0}^{x_1 x_2}$.
In this scenario, there are four space-time points, which can be represented schematically as follows (with the time direction oriented upwards):
        \begin{equation} \label{eq:4point}
        \begin{aligned}
		\begin{tikzpicture}
                \draw[line width=.6pt,black] (0.4,-0.4) -- (0.4,0);	
			\draw[line width=.6pt,black] (-0.4,-0.4) -- (-0.4,0);
              \draw[line width=.6pt,black] (0.4,0.5) -- (0.4,0.9);	
			\draw[line width=.6pt,black] (-0.4,0.5) -- (-0.4,0.9);
                \draw[line width=0.6 pt, fill=gray, fill opacity=0.2] 
		(-0.6,0) -- (.6,0) -- (.6,0.5) -- (-.6,0.5) -- cycle; 
			\node[ line width=0.6pt, dashed, draw opacity=0.5] (a) at (0.6,1.3){$(x_2,t_1)$};
            			\node[ line width=0.6pt, dashed, draw opacity=0.5] (a) at (-0.6,1.3){$(x_1,t_1)$};
			\node[ line width=0.6pt, dashed, draw opacity=0.5] (a) at (-0.4,-0.4){$\bullet$};
                \node[ line width=0.6pt, dashed, draw opacity=0.5] (a) at (0.4,-0.4){$\bullet$};
             \node[ line width=0.6pt, dashed, draw opacity=0.5] (a) at (-0.4,0.9){$\bullet$};
                      \node[ line width=0.6pt, dashed, draw opacity=0.5] (a) at (0.4,0.9){$\bullet$};
			\node[ line width=0.6pt, dashed, draw opacity=0.5] (a) at (-0.6,-0.8){$(x_1,t_0)$};
    	\node[ line width=0.6pt, dashed, draw opacity=0.5] (a) at (0.6,-0.8){$(x_2,t_0)$};
    		\node[ line width=0.6pt, dashed, draw opacity=0.5] (a) at (0,0.25)
      {$\mathcal{E}_{t_1\leftarrow t_0}$};
		\end{tikzpicture}
	\end{aligned}
    \end{equation}
The corresponding right spatiotemporal KD quasiprobability distribution can then be defined as
\begin{equation}
\begin{aligned}
    &\QR(b_1,c_1,b_0,c_0) \\
    = & \Tr \Big[ \mathcal{E}_{t_1\leftarrow t_0} \big( \rho_{t_0}^{x_1 x_2} \, \Pi_{b_0}^{t_0} \otimes \Pi_{c_0}^{t_0} \big) \, \Pi_{b_1}^{t_1} \otimes \Pi_{c_1}^{t_1} \Big].
\end{aligned}
\end{equation}
The distributions $\QL$ and $\QLR$ are defined analogously.

An interesting class of marginals are those that mix space and time, for example $\QR(b_1,c_0)$, where $b_1$ corresponds to $(x_1,t_1)$ while $c_0$ corresponds to $(x_2,t_0)$. 
Due to the presence of evolution, such marginals can also exhibit non-classicality. 
If the evolution is trivial (i.e., the identity channel), then $\QR(b_1,c_0)$ reduces to a proper probability distribution, and no quantumness appears.  
When the evolution is non-trivial, the two-time observable 
\begin{equation}
    \overrightarrow{M}_{b_1,c_0} = \mathcal{E}_{t_1\leftarrow t_0}^{\dagger} (\Pi_{b_1}^{t_1}) \, \Pi_{c_0}^{t_0}
\end{equation}
can induce quantumness, since the evolution typically spreads the support of operators. 
Namely, although $\Pi_{b_1}^{t_1}$ initially acts non-trivially only on the first particle, 
the back-evolved operator $\mathcal{E}_{t_1\leftarrow t_0}^{\dagger} (\Pi_{b_1}^{t_1})$ generally acts non-trivially on the second particle as well, thereby introducing quantumness into the quasiprobability distribution.

Also notice that in the four-spacetime-event setting of Eq.~\eqref{eq:4point}, if we restrict attention to the first qubit and extract the corresponding KD distribution, this goes beyond the purely temporal case introduced earlier, since the evolution in this setting may not be a CPTP map. More precisely, the reduced evolution is a CPTP map only when the initial state is a product state and the original evolution is unitary. For a general setting, however, neither of these conditions holds. Therefore, the marginals of the spatiotemporal KD quasiprobability distribution extend beyond those of the purely temporal KD quasiprobability distribution.

\begin{remark}[Remark on spatiotemporal Kirkwood–Dirac quasiprobability from process matrix]

A general spatiotemporal KD quasiprobability distribution can be obtained from the process tensor \cite{Milz2021quantum} or quantum comb framework \cite{Chiribella2009comb}  (this can also be extended to the process-matrix formalism, which allows for indefinite causal order \cite{oreshkov2012quantum}, although in this work we restrict ourselves to the case of fixed causal order).
As discussed in the main text, this can be obtained based on temporal Born rule.

Using the Stinespring dilation, we can consider the general spatiotemporal setting shown in Figure~\ref{fig:QC}, where the light blue dots in the quantum circuit represent spacetime points. By choosing an arbitrary subset of these spacetime points on which to perform measurements, one can obtain the corresponding spatiotemporal KD quasiprobability distribution. Since the initial state may be entangled and certain local degrees of freedom are traced out, this framework goes beyond the purely temporal setting. In this way, we arrive at the most general form of spatiotemporal KD quasiprobability distributions, with all previously discussed cases appearing as special instances. This will serve as a crucial tool for investigating spatiotemporal properties in large quantum circuits. A more detailed discussion of this perspective will be presented elsewhere.
\end{remark}

\section{Temporal characteristic function and interferometric measurement scheme for temporal Kirkwood-Dirac quasiprobability distributions}
\label{sec:chara}

There exist several approaches for measuring the standard (non-temporal) KD quasiprobability distribution~\cite{Lostaglio2023kirkwooddirac,Gherardini2024KD}, including the weak-value measurement scheme, the cloning scheme, the block-encoding scheme, and the interferometric measurement scheme. In the following, we focus on the temporal generalization of the interferometric measurement scheme, which relies on the characteristic function of the temporal KD quasiprobability distribution.

The characteristic function is a fundamental tool for describing both probability and quasiprobability distributions. 
For the temporal KD quasiprobability distribution, one can similarly define the corresponding characteristic function via a Fourier transform. 
In this way, the right, left, and doubled temporal KD quasiprobability distributions each have their respective temporal characteristic functions. 
Using the interferometric measurement scheme, one can experimentally obtain these temporal characteristic functions, and then apply the inverse Fourier transform to reconstruct the corresponding temporal KD quasiprobability distributions.

For right temporal KD quasiprobability distribution $\QR$, one defines
\begin{equation} \label{eq:KDcharac}
\begin{aligned}
    &\overrightarrow{\chi}_{\rm KD}(u_n,\dots,u_0) 
    = \mathcal{F}\!\left(\QR(b_n,\dots,b_0)\right) \\[0.2cm]
    &= \sum_{b_n,\dots,b_0} \QR(b_n,\dots,b_0)\,
       e^{-i (b_n u_n + \cdots + b_0 u_0)} \\[0.2cm]
    &= \Tr \!\Big[
        \mathcal{E}_{t_n \leftarrow t_{n-1}}
        \Big(
            \cdots 
            \mathcal{E}_{t_1 \leftarrow t_0}\!\big(\rho_{t_0} e^{-i B_0 u_0}\big)
            e^{-i B_1 u_1} \cdots
        \Big)
        e^{-i B_n u_n}
    \Big],
\end{aligned}
\end{equation}
where $\mathcal{F}$ denotes the Fourier transform.
We refer to this quantity as the right \emph{temporal KD characteristic function}. 
Taking the inverse Fourier transform then yields right KD quasiprobability distribution $\QR$, which implies that the temporal KD characteristic function $\overrightarrow{\chi}_{\rm KD}$ contains the full information of $\QR$. In other words, if $\overrightarrow{\chi}_{\rm KD}$ is fully determined, $\QR$ can also be fully reconstructed.

\begin{figure}
    \centering
    \includegraphics[width=7cm]{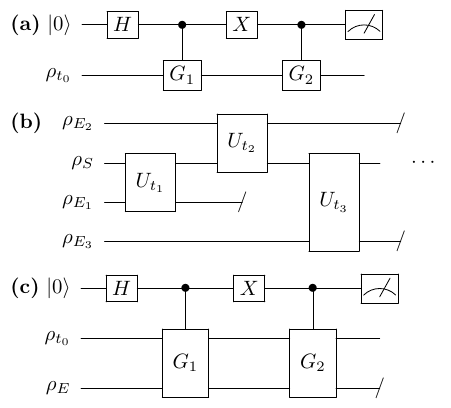}
  \caption{Quantum circuits illustrating the interferometric measurement scheme. 
(a) The scheme for unitary evolution. 
(b) The Stinespring dilation of a CPTP evolution. 
(c) The scheme for CPTP evolution.}
    \label{fig:KD-QC}
\end{figure}

Similarly, one can define right temporal KD characteristic function $\overleftarrow{\chi}_{\rm KD}$ from  right KD quasiprobability distribution $\QL$ as
\begin{equation} \label{eq:KDcharacL}
\begin{aligned}
    &\overleftarrow{\chi}_{\rm KD}(v_n,\dots,v_0) 
    = \mathcal{F}\!\left(\QL(a_n,\dots,a_0)\right) \\[0.2cm]
    &= \sum_{b_n,\dots,b_0} \QL(b_n,\dots,b_0)\,
       e^{i (a_n v_n + \cdots + a_0 v_0)}   \\[0.2cm]
    &= \Tr \!\Big[
         e^{i A_n v_n} \mathcal{E}_{t_n \leftarrow t_{n-1}}
         \Big(
            \cdots 
           e^{i A_1 v_1}  \mathcal{E}_{t_1 \leftarrow t_0}\!\big(e^{i A_0 v_0} \rho_{t_0} \big)
            \cdots
        \Big)
    \Big].
\end{aligned}
\end{equation}
where we intentionally choose different signs in the exponent to ensure that 
$\overleftarrow{\chi}_{\rm KD}^* = \overrightarrow{\chi}_{\rm KD}$ when the same measurement settings are used at each time step. 
This convention is convenient when discussing temporal states based on temporal KD quasiprobability distributions.

The doubled characteristic function $\overleftrightarrow{\chi}_{\rm KD}$ can be obtained from doubled  KD quasiprobability distribution $\QLR$
    \begin{equation} \label{eq:KDcharacIC}
\begin{aligned}
    &\overleftrightarrow{\chi}_{\rm KD}(v_n,\dots,v_0;u_n,\cdots,u_0) 
    = \mathcal{F}\!\left(\QLR(a_n,\dots,a_0;b_n,\cdots,b_0)\right) \\[0.2cm]
    = & \sum_{\substack{a_0,\ldots,a_n \\ b_0,\ldots,b_n}} \QLR(a_n,\dots,a_0;b_n,\cdots,b_0)\,
       e^{i (a_n v_n + \cdots + a_0 v_0)} e^{-i (b_n u_n + \cdots + b_0 u_0)} \\[0.2cm]
    = & \Tr \!\Big[
      e^{i A_n v_n}   \mathcal{E}_{t_n \leftarrow t_{n-1}}
        \Big(
            \cdots 
          e^{i A_1 v_1}   \mathcal{E}_{t_1 \leftarrow t_0}\!\big(e^{i A_0 v_0} \rho_{t_0} e^{-i B_0 u_0}\big)
            e^{-i B_1 u_1} \cdots
        \Big)
        e^{-i B_n u_n}
    \Big].
\end{aligned}
\end{equation}
We define the temporal KD characteristic functions 
$\overleftrightarrow{\chi}_{\rm KD}$ so that they connect naturally with $\overleftarrow{\chi}_{\rm KD}$ and $\overrightarrow{\chi}_{\rm KD}$. In summary, we have the following result:
\begin{lemma}
The temporal KD characteristic functions satisfy the following properties:

(i) The left and right temporal KD characteristic functions are related by
\begin{equation}
    \overleftarrow{\chi}_{\rm KD}^* = \overrightarrow{\chi}_{\rm KD},
\end{equation}
when the same measurement settings are used at each time step.

(ii) The left, right, and doubled temporal KD characteristic functions are related by
\begin{equation}
\begin{aligned}
    \overleftrightarrow{\chi}_{\rm KD}(0, \dots, 0; u_n, \dots, u_0) &= \overrightarrow{\chi}_{\rm KD}(u_n, \dots, u_0), \\
    \overleftrightarrow{\chi}_{\rm KD}(v_n, \dots, v_0; 0, \dots, 0) &= \overleftarrow{\chi}_{\rm KD}(v_n, \dots, v_0).
\end{aligned}
\end{equation}

(iii) The temporal MH characteristic function (if we adopt a definition similar to Eq.~\eqref{eq:KDcharac}, i.e., 
\(
\chi_{\rm MH} (w_n, \dots, w_0) = \mathcal{F}(Q_{\rm MH}) = \sum_{b_n,\dots,b_0} Q_{\rm MH}(b_n, \dots, b_0)\, 
e^{-i (b_n w_n + \cdots + b_0 w_0)} 
\))
satisfies
\begin{equation}
\begin{aligned}
    &\chi_{\rm MH} (w_n, \dots, w_0) \\
    = &\frac{1}{2} (
\overrightarrow{\chi}_{\rm KD}(w_n, \dots, w_0) + \overleftarrow{\chi}_{\rm KD}(-w_n, \dots, -w_0)).
\end{aligned}  
\end{equation}
All of these results can be readily verified from the definitions.
Note also that once the doubled KD characteristic function is obtained, all other characteristic functions can be derived from it.
\end{lemma}

The multi-time temporal KD characteristic function can be measured experimentally by a slight extension of the interferometric measurement scheme \cite{Mazzola2013Characterisitc,Dorner2013KDchara}.  
By taking its inverse Fourier transform, one then obtains the temporal KD distribution.

To simplify the discussion and build intuition, we begin by considering unitary evolution at each time step,
$\mathcal{E}_{t_k \leftarrow t_{k-1}}(\bullet) = U_{t_k} (\bullet) U_{t_k}^{\dagger}$.
The interferometric measurement scheme  can be implemented using the quantum circuit in Figure~\ref{fig:KD-QC} (a), 
where the unitary gates are defined as
$G_1 = U_{t_n} \cdots U_{t_1}$
and
$G_2 = e^{i B_n u_n} U_{t_n} e^{i B_{n-1} u_{n-1}} \cdots U_{t_2} e^{i B_1 u_1} U_{t_1} e^{i B_0 u_0}$. Notice that the choices of $G_1$ and $G_2$ are not unique; they are constructed so as to ensure that we obtain the temporal characteristic function in Eq.~\eqref{eq:KDcharac}.
By performing measurements in the Pauli-$X$ basis and Pauli-$Y$ basis for the first qubit, one can extract the real and imaginary parts, $\Re \overrightarrow{\chi}_{\rm KD}$ and $\Im \overrightarrow{\chi}_{\rm KD}$, respectively.

In an analogous manner, the temporal characteristic functions $\overleftarrow{\chi}_{\rm KD}$ and $\overleftrightarrow{\chi}_{\rm KD}$ can be obtained via replacing the gates in Figure~\ref{fig:KD-QC} (a) as follows: 
(i) For left characteristic function $\overleftarrow{\chi}_{\rm KD}$ two gates are chosen as $G_1 = e^{i A_n v_n} U_{t_n} e^{i A_{n-1} v_{n-1}} \cdots U_{t_2} e^{i A_1 v_1} U_{t_1} e^{i A_0 v_0}$ and $G_2= U_{t_n} \cdots U_{t_1}$; 
(ii) For doubled characteristic function $\overleftrightarrow{\chi}_{\rm KD}$, two gates are chosen as
$G_1 = e^{i A_n v_n} U_{t_n} e^{i A_{n-1} v_{n-1}} \cdots U_{t_2} e^{i A_1 v_1} U_{t_1} e^{i A_0 v_0}$
and
$G_2 = e^{i B_n u_n} U_{t_n} e^{i B_{n-1} u_{n-1}} \cdots U_{t_2} e^{i B_1 u_1} U_{t_1} e^{i B_0 u_0}$.
It is straightforward to verify that by performing measurements in the Pauli-$X$ and Pauli-$Y$ bases on the first qubit, one can extract the real and imaginary parts of the corresponding temporal characteristic functions.

The temporal KD characteristic function for general CPTP evolutions can also be measured by introducing an ancilla system, based on the Stinespring dilation theorem.  
Taking $\overrightarrow{\chi}_{\rm KD}$ as an example, for each time evolution $\mathcal{E}_{t_k \leftarrow t_{k-1}}$, we introduce an ancilla qubit $|0\rangle_{E_k}$ and a unitary $U^{SE_k}_{t_k}$ such that
\[
\Tr_{E_k} \Big[ U^{SE_k}_{t_k} \big(\rho_{S} \otimes |0\rangle\langle 0|_{E_k} \big) (U^{SE_k}_{t_k})^{\dagger} \Big] = \mathcal{E}_{t_k \leftarrow t_{k-1}}(\rho_{S}).
\]
In this way, the multi-time evolution can be fully characterized by the corresponding unitary evolutions on the system and ancilla qubits, see Figure~\ref{fig:KD-QC} (b).

The characteristic function can also be expressed in terms of the Stinespring dilation as follows:
\begin{equation}
    \begin{aligned}
    \overrightarrow{\chi}_{\rm KD}(u_n,\dots,u_0) 
    = \Tr \!\Big[
        U_{t_1}^{SE_n} 
        \Big(
            \cdots 
            U_{t_1}^{SE_1} \!\big((\rho_{t_0}\otimes \rho_E ) (e^{-i B_0 u_0} \otimes \mathds{I}) \big) (U_{t_1}^{SE_1})^{\dagger} 
            (e^{-i B_1 u_1} \otimes \mathds{I})  \cdots
        \Big) (U_{t_1}^{SE_1} )^{\dagger}
        (e^{-i B_n u_n} \otimes \mathds{I})
    \Big].
\end{aligned}
\end{equation}
Here, we define $\rho_E = \bigotimes_{i=1}^n |0\rangle \langle 0|_{E_i}$, and all unitary gates $e^{-i B_k u_k}$ act on the system qudit. It is straightforward to verify that the resulting function is consistent with Eq.~\eqref{eq:KDcharac}.

Based on the above observation, the characteristic function for a general CPTP evolution can be experimentally measured using the same circuit as for unitary evolution.  
The corresponding characteristic function can thus be obtained using the quantum circuit in Figure~\ref{fig:KD-QC} (c),
where 
$G_1 = U_{t_n}^{SE_n} \cdots U_{t_1}^{SE_1}$
and
$G_2 = (e^{i B_n u_n} \otimes \mathds{I}) U_{t_n}^{SE_n} (e^{i B_{n-1} u_{n-1}}\otimes \mathds{I}) \cdots U_{t_2}^{SE_2} (e^{i B_1 u_1}\otimes \mathds{I}) U_{t_1}^{SE_1} (e^{i B_0 u_0}\otimes \mathds{I})$.

In an analogous manner, the temporal KD characteristic functions $\overleftarrow{\chi}_{\rm KD}$ and $\overleftrightarrow{\chi}_{\rm KD}$ can be obtained: (i) For $\overleftarrow{\chi}_{\rm KD}$ two gates are chosen as
$G_1 = (e^{i A_n v_n} \otimes \mathds{I}) U_{t_n}^{SE_n} (e^{i A_{n-1} v_{n-1}}\otimes \mathds{I}) \cdots U_{t_2}^{SE_2} (e^{i A_1 v_1}\otimes \mathds{I}) U_{t_1}^{SE_1} (e^{i A_0 v_0}\otimes \mathds{I})$ and
$G_2 = U_{t_n}^{SE_n} \cdots U_{t_1}^{SE_1}$; 
(ii) For doubled $\overleftrightarrow{\chi}_{\rm KD}$, two gates are chosen as
$G_1 = (e^{i A_n v_n} \otimes \mathds{I}) U_{t_n}^{SE_n} (e^{i A_{n-1} v_{n-1}}\otimes \mathds{I}) \cdots U_{t_2}^{SE_2} (e^{i A_1 v_1}\otimes \mathds{I}) U_{t_1}^{SE_1} (e^{i A_0 v_0}\otimes \mathds{I})$
and
$G_2 = (e^{i B_n u_n} \otimes \mathds{I}) U_{t_n}^{SE_n} (e^{i B_{n-1} u_{n-1}}\otimes \mathds{I}) \cdots U_{t_2}^{SE_2} (e^{i B_1 u_1}\otimes \mathds{I}) U_{t_1}^{SE_1} (e^{i B_0 u_0}\otimes \mathds{I})$.
It is straightforward to verify that the measurement probabilities indeed yield the corresponding temporal KD characteristic functions.

If the doubled KD quasiprobability distribution is obtained, one can derive from it the left and right KD distributions, the LvN distribution, as well as the left/right, and doubled MH distributions.

\section{Weak values and temporal Kirkwood-Dirac quasiprobability distribution}
\label{app:weak}

The temporal KD quasiprobability distribution can be accessed experimentally in several different ways~\cite[Section~4]{Lostaglio2023kirkwooddirac}, the traditional approach being the weak-value protocol~\cite{ArvidssonShukur2024KDreview,Lostaglio2023kirkwooddirac}.

Recall that the weak value~\cite{Aharonov1988weakvalue} of an observable $A$ with respect to a pre-selected initial state $|\varphi_i\rangle$ and a post-selected state $|\xi_j\rangle$ is defined as  
\begin{equation}
    A_{\rm weak}(\xi_j,\varphi_i) := 
    \frac{\langle \xi_j|A|\varphi_i\rangle}{\langle \xi_j|\varphi_i\rangle}.
\end{equation}
If we restrict to a unitary channel such that $\mathcal{E}(\Pi_{b_1}^{t_1}) = |\xi_{b_1}^{t_1}\rangle \langle \xi_{b_1}^{t_1}|$ for some state $|\xi_{b_1}^{t_1}\rangle$, and consider the observable $B_0 = \sum_{b_0} b_0 \Pi_{b_0}^{t_0}$, then the weak value coincides with the average over the conditional temporal KD quasiprobability distribution:  
\begin{equation}
    \langle B_0 \rangle_{\rm weak} 
    = \sum_{b_0} b_0 \,\frac{\QR(b_1,b_0)}{p_{t_1}(b_1)}.
\end{equation}
This implies that the conditional KD quasiprobabilities can be interpreted as weak values.

The connection between the KD quasiprobability distribution and weak-value measurements has several intriguing applications, including quantum contextuality~\cite{Lostaglio2020KDcontextual}, Leggett-Garg test \cite{Dressel2011weakLG} and quantum metrology~\cite{Pusey2014weakContextual,ArvidssonShukur2020KDmetrology}.

\section{Proof of classicality criterion}
\label{app:KNnogo}

The following result is well known and useful (see, e.g., \cite[Theorem~1.3.1]{ludwig1983foundations}, \cite[Proposition~1]{Busch2014RMPmeasure}).  
We formulate it in a manner suited to our setting and provide an elementary proof in the finite-dimensional case for convenience.
\begin{lemma} \label{lemma:commu}
  For a given set of positive semidefinite operators $M_{ij}$,  if $\sum_i M_{ij}=B_j$ and $\sum_{j}M_{ij}=A_i$ and $\{A_i\}$ is complete orthogonal projectors, viz. $A_i^{\dagger}=A_i$ and $A_i^2=A_i$ for all $i$ and $A_iA_k=\delta_{ik}A_i$, $\sum_i A_i=\mathds{I}$,  then $[A_i,B_j]=0$ and $M_{ij}=A_iB_j$.
\end{lemma}

\begin{proof}
Fix $i$ and set $P := I - A_i$. Then
\[
 P \Big(\sum_j M_{ij}\Big) P = \sum_j P M_{ij} P = P A_i P = 0.
\]
Positivity of $M_{ij}$ implies $P M_{ij} P$ is positive for all $i,j$, and the above equality then forces $P M_{ij} P = 0$ for all $i,j$ (since the sum of positive operators is zero only if each term is zero).  
This gives $\sqrt{M_{ij}} P = 0 = P \sqrt{M_{ij}}$, and multiplying by $\sqrt{M_{ij}}$ yields $P M_{ij} = M_{ij} P = 0$.  
Hence
\[
\begin{aligned}
    M_{ij} &= (A_i + P) M_{ij} (A_i + P) \\
           &= A_i M_{ij} A_i + A_i M_{ij} P + P M_{ij} A_i + P M_{ij} P \\
           &= A_i M_{ij} A_i.
\end{aligned}
\]
Now compute
\[
\begin{aligned}
   A_i B_j =& A_i \sum_k M_{kj} = A_i \sum_k A_k M_{kj} A_k  \\
   = & \sum_k \delta_{ik} A_k M_{kj} A_k = A_i M_{ij} A_i = M_{ij}, 
\end{aligned}
\]
where we used the orthogonality $A_i A_k = \delta_{ik} A_k$.
A similar calculation shows $B_j A_i = M_{ij}$. Therefore,
\[
M_{ij} = A_i B_j = B_j A_i \quad \text{and} \quad [A_i, B_j] = 0,
\]
we thus arrive at the required conclusion.
\end{proof}

The following extension is also a well-known result in quantum information theory. For completeness, we also provide a proof in the form that will be used in this work:

\begin{lemma}[Extension of functionals on density operators]
Let $\mathcal H$ be a finite-dimensional Hilbert space and let 
$\mathsf{Herm}(\mathcal{H})$ be the real vector space of all Hermitian operators and 
$\mathsf{State}(\mathcal{H}) \subset \mathsf{Herm}(\mathcal{H})$ as the set of all density operators.
Suppose $f:\mathsf{State}(\mathcal{H})  \to \mathbb R$ is \emph{convex-linear}, i.e., for all $\rho_1, \rho_2 \in \mathcal D$ and $0\le \lambda \le 1$,
\[
f(\lambda \rho_1 + (1-\lambda)\rho_2) = \lambda f(\rho_1) + (1-\lambda) f(\rho_2).
\]
Then there exists a \emph{unique linear functional} $L: \mathsf{Herm}(\mathcal H) \to \mathbb R$ such that
\[
L(\rho) = f(\rho), \quad \forall \rho \in \mathsf{State}(\mathcal{H}).
\]
Equivalently, there exists a \emph{unique Hermitian operator} $M \in \mathsf{Herm}(\mathcal H)$ such that
\[
f(\rho) = \operatorname{Tr}(M \rho), \quad \forall \rho \in \mathsf{State}(\mathcal{H}).
\]
\end{lemma}

\begin{proof}
We prove this with several steps.

\emph{1. $\mathsf{Herm}(\mathcal{H})$ has a basis consisting of density operators.}
It is straightforward to see that there exists a finite set of projectors $\{P_i\}$ forming a basis of $\mathsf{Herm}(\mathcal{H})$. Upon appropriate normalization, each projector defines a density operator
\[
\chi_i = \frac{P_i}{\Tr P_i}.
\]
The set $\{\chi_i\}$ therefore also forms a basis of $\mathsf{Herm}(\mathcal{H})$.

\emph{2. The extension exists and is unique.} 
We define $L(\chi_i) = f(\chi_i)$ on a basis $\{\chi_i\} \subset \mathsf{State}(\mathcal{H})$, and then extend $L$ linearly to the entire space $\mathsf{Herm}(\mathcal{H})$. 
Convex-linearity ensures that this extension satisfies $L = f$ when restricted to $\mathsf{State}(\mathcal{H})$.

\emph{3. There exists a matrix $M$ representing the functional $L$.} 
This follows from the Riesz representation theorem. We take $M$ to represent $L$ with respect to the Hilbert–Schmidt inner product 
\[
\langle A, B \rangle := \Tr(A^\dagger B),
\] 
which reduces to $\langle A, B \rangle = \Tr(AB)$ on the space of Hermitian operators. Since $L$ is a bounded linear functional, there exists a matrix $M$ such that 
\[
L(X) = \langle M, X \rangle \quad \text{for all } X \in \mathsf{Herm}(\mathcal{H}).
\]
We emphasize that the set $\{f(\chi_i)=\langle M, \chi_i \rangle\}_i$ completely determines $M$, which implies that $f$ fully determines $M$.
\end{proof}

To prove classicality criterion, we first note that the probability functional over density operators can be extended to the entire space $\mathsf{Herm}(\mathcal{H})$. This yields an operator $M_{b_n\cdots b_0}$ such that
\begin{equation}
    p(b_n,\ldots,b_0) = \Tr\!\bigl(\rho_{t_0} M_{b_n\cdots b_0}\bigr).
\end{equation}
Since the correct marginals are obtained for any initial state (and in particular for the basis density operators), comparing the two expressions for the probabilities gives
\begin{equation}
    M_{b_k} = \sum_{b_0, \ldots, \widehat{b_k}, \ldots, b_n} M_{b_n\cdots b_0} 
    = \mathcal{E}^{\dagger}_{t_1 \leftarrow t_0} \circ \cdots \circ \mathcal{E}^{\dagger}_{t_k \leftarrow t_{k-1}} \bigl(\Pi_{b_k}^{t_k}\bigr),
\end{equation}
for all time steps $t_k$, where $\widehat{b_k}$ indicates summation over all indices except $b_k$.  
Noting that $M_{b_0} = \Pi_{b_0}^{t_0}$ is a projector, Lemma~\ref{lemma:commu} then implies
\begin{equation}
    [M_{b_n\cdots b_1}, M_{b_0}] = 0, \qquad 
    M_{b_n\cdots b_0} = M_{b_n\cdots b_1} M_{b_0}.
\end{equation}
This establishes the first statement of classicality criterion.

When all evolutions are unitary, each $M_{b_k}$ is a projector for every time step $t_k$. 
By repeated application of Lemma~\ref{lemma:commu}, the second statement of classicality criterion then follows.

\begin{example}[Replacement channel]
For the replacement channel $\mathcal{R}_{\omega} := \omega \Tr(\bullet)$ with $\omega$ a density operator, the temporal KD quasiprobability distribution factorizes over different times:
\begin{align}
  &  \QR(b_1,b_0) = p_{t_1}(b_1)\, p_{t_0}(b_0), \\
   & \QLR(a_1,a_0;b_1,b_0) = Q^{t_1}_{\rm KD}(a_1,b_1)\, Q^{t_0}_{\rm KD}(a_0,b_0),
\end{align}
where $p$ denotes the LvN measurement statistics, and $Q_{\rm KD}$ denotes the standard KD quasiprobability distribution.  
This product form of temporal KD quasiprobability distribution reflects the fact that the replacement channel breaks temporal correlations.
The KD non-classicality is of the form:
\begin{align}
   & \mathcal{N}[\QR]=0,\\
   & \mathcal{N}[ \QLR] =  \mathcal{N}[Q_{KD}^{t_1}] \mathcal{N}[Q_{KD}^{t_0}] +\mathcal{N}[Q_{KD}^{t_1}] +\mathcal{N}[Q_{KD}^{t_0}].
\end{align}
Therefore, the $\QR$ (or equivalently $\QL$) plays a crucial role in detecting temporal KD non-classicality.

In the multi-time case, if all evolutions are given by replacement channels, the distribution retains this product form.  
Thus, the left and right temporal KD quasiprobabilities exhibit no non-classicality. In contrast, the doubled temporal KD quasiprobability may display non-classicality, arising from the state and measurement choice at each time step.  
\end{example}

\begin{example}[Measure-and-replace channel]
Consider a measure-and-replace quantum channel (also called entanglement-breaking channel) that performs a measurement $\mathcal{M}_k$ (which is a trace-nonincreasing CP map such that $\sum_k \mathcal{M}_k$ is a CPTP map) on the input state and then replaces the post-measurement state with a fixed output state $\omega_k$ depending on the measurement outcome. The action of this channel can be written as
\begin{equation}
    \mathcal{E}(\rho) = \sum_k \operatorname{Tr}[\mathcal{M}_k(\rho)]\, \omega_k.
\end{equation}
This map is CPTP by construction and the replacement channel can be regarded as a special case of this channel.
Physically, it corresponds to a process in which the input system is measured, and the resulting classical information is used to prepare a corresponding output quantum state. 
The corresponding right two-time KD quasiprobability distribution is given by
\begin{equation}
    \QR(b_1,b_0)= \sum_k p^{t_1}(b_1|\omega_k)\, Q^{t_0}_{\mathrm{KD}}(b_0,k),
\end{equation}
where $p^{t_1}(b_1|\omega_k)=\operatorname{Tr}(\omega_k \Pi_{b_1}^{t_1})$ denotes the measurement statistics at time step $t_1$, and 
$Q^{t_0}_{\mathrm{KD}}(b_0,k)=\operatorname{Tr}[\mathcal{M}_k(\rho \Pi_{b_0}^{t_0})]$ represents the extended KD distribution at time step $t_0$, which serves as the origin of quantumness in this process:
\begin{equation}
\begin{aligned}
     \mathcal{N}[\QR(b_1,b_0)] = & \sum_{k,b_0} |Q^{t_0}_{\mathrm{KD}}(b_0,k)|-1 \\
     =&\mathcal{N}[Q^{t_0}_{\mathrm{KD}}(b_0,k)].
\end{aligned}
\end{equation}
This matches our intuition, as this channel does not introduce any quantum correlations; consequently, the temporal KD quantumness receives no contribution from the evolution itself.
The quantumness of the extended KD distribution at time step $t_0$ in fact originates from the measurement process 
$\mathcal{M}_k(\rho)=\sum_j E_{k,j} \rho E_{k,j}^{\dagger}$ with $E_{k,j}^{\dagger}$ are Kraus operators.
We refer to it as an extended KD distribution because when $\mathcal{M}_k$ is chosen to be a projective measurement, i.e., $\mathcal{M}_k(\bullet)=\Pi_k(\bullet)\Pi_k$, it reduces to the standard KD distribution 
$Q_{\mathrm{KD}}^{t_0}(b_0,k)=\operatorname{Tr}(\rho\, \Pi_{b_0}\Pi_k)$.

The double temporal KD quasiprobability distribution also takes a product form,
\begin{equation}
\begin{aligned}
   & \QLR(a_1,a_0;b_1,b_0) \\
    = &\sum_k Q_{\rm KD}^{t_0}(a_0,b_0,k)\, Q_{\rm KD}^{t_1}(a_1,b_1|\omega_k),
    \end{aligned}
\end{equation}
where we have $Q_{\rm KD}^{t_0}(a_0,b_0,k)=\Tr \mathcal{M}_k (\Pi_{a_0} \rho \Pi_{b_0})$ and $Q^{t_1}_{\rm KD}(a_1,b_1|\omega_k)=\Tr (\Pi_{a_1} \omega_k \Pi_{b_1})$.
The quantumness still arises independently from the two time steps, with no contribution from the evolution itself.

\end{example}

\section{Unification of spatiotemporal states through spatiotemporal Bloch tomography}

Recall that the temporal state in the PDO formalism \cite{fitzsimons2015quantum,jia2023quantumspace,liu2023unification,Liu2025PDO,Song2024PDO,Parzygnat2023pdo} is obtained by extending spatial tomography into the temporal domain. For example, for a two-qubit state, we may measure all joint Pauli observables
\begin{equation}
   T^{\mu,\nu} = \langle \sigma_{\mu}\otimes \sigma_{\nu} \rangle 
   = \Tr \big[(\sigma_{\mu}\otimes \sigma_{\nu}) \rho \big].
\end{equation}
The state can then be expressed as $\rho = \frac{1}{2^2}\sum_{\mu,\nu} T^{\mu,\nu}\, \sigma_{\mu}\otimes \sigma_{\nu}$. Physically, the correlator $\langle \sigma_{\mu}\otimes \sigma_{\nu} \rangle$ is obtained by measuring all Pauli operators and reconstructing the joint probability distribution $p(a,b|\sigma_{\mu},\sigma_{\nu})$, so that 
$\langle \sigma_{\mu}\otimes \sigma_{\nu} \rangle = \sum_{a,b} ab\, p(a,b|\sigma_{\mu},\sigma_{\nu})$. 
In a two-time temporal setting, we can proceed analogously by implementing sequential Pauli measurements and obtaining the temporal joint LvN probabilities $p(a,b|\sigma_{\mu}^{t_1},\sigma_{\nu}^{t_0})$. From these, we compute the temporal correlators
\begin{equation}
   T^{\mu,\nu} = \langle \{\sigma_{\mu}^{t_1}, \sigma_{\nu}^{t_0} \}\rangle
   = \sum_{a,b} ab\, p(a,b|\sigma_{\mu}^{t_1},\sigma_{\nu}^{t_0}).
\end{equation}
We can then write down a PDO by borrowing the same expression as in spatial tomography,
\begin{equation}
   R = \frac{1}{2^2}\sum_{\mu,\nu} T^{\mu,\nu}\, \sigma_{\mu}\otimes \sigma_{\nu}.
\end{equation}
The resulting PDO is Hermitian but, in general, not positive semidefinite. 
In this work, we will also refer to the state obtained through temporal tomography as a \emph{temporal state}. 
As we will see, depending on the chosen tomography protocol, such a state may go beyond the PDO formalism.

Unlike spatial density operators, the PDO does not obey a (generalized) Born rule. In general, the LvN distribution does not coincide with the that obtained from Born rule
\begin{equation}
    p(a,b|\sigma_{\mu}^{t_1},\sigma_{\nu}^{t_0}) \neq \Tr \big[(\Pi_{a}^{t_1}\otimes \Pi_{b}^{t_0}) R\big].
\end{equation}
Nevertheless, it yields correct joint expectation values only for spin operators of the form $A \otimes B$, where $A=\Vec{a}\cdot \Vec{\sigma}$ and $B=\Vec{b}\cdot \Vec{\sigma}$ (or, more generally, for ``light-touch operators'') \cite{fullwood2025operatorrep}:
\begin{equation}
   \langle \{A,B\}\rangle = \Tr \big[(A\otimes B) R\big].
\end{equation}
More recently, the relationship between the temporal Born rule and the MH quasiprobability distribution has been investigated in the context of two-time PDOs \cite{fullwood2025spatiotemporalbornrule}.

In this part, we introduce an alternative approach to obtain the temporal state via tomography based on the temporal KD  and MH quasiprobability distributions, rather than using the LvN distribution. 
The KD and MH quasiprobability distributions can be experimentally accessed either through direct measurement (Section~\ref{sec:chara}) or via weak measurement (Section~\ref{app:weak} in Appendix); also see \cite{ArvidssonShukur2024KDreview,Lostaglio2023kirkwooddirac} for a detailed discussion of other methods. 
To implement temporal tomography, we select a tomographically complete set of observables such as Pauli matrices, which provides access to the joint expectation values across multi-time steps, from which the temporal state can be reconstructed. For a spatiotemporal setting, the construction can be extended straightforwardly.

For a local system of dimension $d$, we consider generalized Pauli operators (also referred to as Hilbert–Schmidt operators), which satisfy the following properties: 
(i) $\sigma_0 = \mathds{I}$; 
(ii) $\Tr(\sigma_j) = 0$ for all $j \geq 1$; 
(iii) orthogonality: $\Tr(\sigma_\mu \sigma_\nu) = d \, \delta_{\mu \nu}$.
These operators form an orthogonal basis for the real vector space of Hermitian operators $\Herm(\mathcal{H})$, with $\mathcal{H} = \mathbb{C}^d$.

Via measuring the left temporal KD quasiprobability distributions, we can compute the corresponding left KD joint expectation values of (generalized) Pauli operators as $\overrightarrow{T}^{\mu_n,\ldots,\mu_0} = \sum_{a_n,\ldots,a_0} a_n\cdots a_0 \, \overrightarrow{Q}_{\rm KD}(a_n,\ldots,a_0 \mid \sigma_{\mu_n},\ldots,\sigma_{\mu_0})$, which can be interpreted as the  KD-type temporal correlators of  Pauli operators across multi-time steps.
Similarly, for the right KD quasiprobability distribution, we have $\overleftarrow{T}^{\mu_n,\ldots,\mu_0} = \sum_{a_n,\ldots,a_0} a_n\cdots a_0 \, \overleftarrow{Q}_{\rm KD}(a_n,\ldots,a_0 \mid \sigma_{\mu_n},\ldots,\sigma_{\mu_0})$. From Lemma~\ref{lemma:QDleftright}, the temporal joint expectation values obtained from the left and right  KD quasiprobability distributions are related via complex conjugation, since all outcomes of Pauli operators are real.
For the doubled temporal KD quasiprobability distribution, we obtain $\overleftrightarrow{T}^{\mu_n,\ldots,\mu_0; \nu_n,\ldots,\nu_0} 
       = \sum_{a_n,\ldots,a_0; b_n,\ldots,b_0} a_n\cdots a_0 \, b_n\cdots b_0 \, \overleftrightarrow{Q}_{\rm KD}(a_n,\ldots,a_0; b_n,\ldots,b_0 \mid \sigma_{\mu_n},\ldots,\sigma_{\mu_0}; \sigma_{\nu_n},\ldots,\sigma_{\nu_0})$.
If we take the left half of observables to be identity operators, then we obtain $\overleftrightarrow{T}^{0,\cdots,0;\nu_n,\cdots, \nu_0}=\overrightarrow{T}^{\nu_n,\ldots,\nu_0}$; similarly, $\overleftrightarrow{T}^{\mu_n,\cdots, \mu_0;0,\cdots,0}=\overleftarrow{T}^{\mu_n,\ldots,\mu_0}$.
Notice that $\overleftrightarrow{T}^{\mu_n,\ldots,\mu_0; \nu_n,\ldots,\nu_0}$ is the same as doubled correlation tensor in Ref.~\cite{jia2024spatiotemporal}, which is equivalent to the process tensor~\cite{Milz2021quantum} in characterizing a temporal quantum process.

From the temporal joint expectation values of Pauli observables, one can construct the temporal state by invoking the Bloch representation of density operators. From right temporal KD tomography we obtain \emph{right KD temporal state}
\begin{equation}
\overrightarrow{\Upsilon}
= \frac{1}{d^{n+1}}
\sum_{\mu_0,\ldots,\mu_n=0}^{d^2-1}
\overrightarrow{T}^{\mu_n,\ldots,\mu_0}
\,\sigma_{\mu_n}\otimes\cdots\otimes\sigma_{\mu_0}.
\end{equation}
Similarly, from left spatiotemporal KD tomography we obtain \emph{left KD temporal state}
\begin{equation}
\overleftarrow{\Upsilon}
= \frac{1}{d^{n+1}}
\sum_{\mu_0,\ldots,\mu_n=0}^{d^2-1}
\overleftarrow{T}^{\mu_n,\ldots,\mu_0}
\,\sigma_{\mu_n}\otimes\cdots\otimes\sigma_{\mu_0}.
\end{equation}
From doubled spatiotemporal KD tomography we obtain \emph{doubled KD temporal state}
\begin{equation}
\overleftrightarrow{\Upsilon}
= \frac{1}{d^{n+1}}
\sum_{\substack{\mu_0,\ldots,\mu_n\\\nu_0,\ldots,\nu_n=0}}^{d^2-1}
\overleftrightarrow{T}^{\mu_n,\ldots,\mu_0;\,\nu_n,\ldots,\nu_0}
\biggl(\bigotimes_{i=0}^n \sigma_{\mu_i}\biggr)\otimes
\biggl(\bigotimes_{j=0}^n \sigma_{\nu_j}\biggr)
\end{equation}
We observe that a tomographically complete set of observables (here, the Hilbert-Schmidt operator basis) plays a crucial role. Alternatively, one may employ other types of operators (e.g., ``light-touch operators''\cite{Liu2025PDO,fullwood2025operatorrep}) to implement temporal tomography and obtain the corresponding temporal state. 
Temporal states obtained from different choices of orthonormal basis (with respect to Hilbert-Schmidt inner product) in the observable space are the same.

For the temporal MH quasiprobability distribution, temporal tomography can be similarly implemented to reconstruct the corresponding temporal state.  
Since the left and right temporal MH quasiprobability distributions coincide, the corresponding MH temporal correlators are given by the real parts of the KD quasiprobability distribution correlators:
\begin{equation}
    T_{\rm MH}^{\mu_n,\cdots,\mu_0} = \operatorname{Re} \overleftarrow{T}^{\mu_n,\cdots,\mu_0} = \operatorname{Re} \overrightarrow{T}^{\mu_n,\cdots,\mu_0}.
\end{equation}
The resulting temporal state is denoted by ${\Upsilon}^{\rm MH}$ and referred to as the \emph{left/right MH temporal state}.  
For the doubled temporal MH quasiprobability distribution, the MH temporal correlators are likewise the real parts of the KD temporal correlators:
\begin{equation}
    T_{\rm MH}^{\mu_n,\cdots,\mu_0;\nu_n,\cdots,\nu_0} = \operatorname{Re} \overleftrightarrow{T}^{\mu_n,\cdots,\mu_0;\nu_n,\cdots,\nu_0}.
\end{equation}
The resulting temporal state is denoted by $\overleftrightarrow{\Upsilon}^{\rm MH}$ and referred to as the \emph{doubled MH temporal state}.

From the definition, it's easy to verify that doubled KD temporal state $\overleftrightarrow{\Upsilon}$ coincides with the doubled density operator \cite{jia2024spatiotemporal}. It is also closely related to other formalisms, such as the quantum comb \cite{Chiribella2009comb}, the process tensor \cite{Milz2021quantum}, quantum strategies \cite{gutoski2007toward}, and process matrices with definite causal order \cite{oreshkov2012quantum}. Under vectorization, $\sigma_{\mu} \mapsto |\sigma_{\mu}\rrangle$, $\overleftrightarrow{\Upsilon}$ can be mapped to the superdensity operator \cite{cotler2018superdensity}. Therefore, temporal KD tomography provides an operational interpretation of all these formalisms, which appear in different disguises.
The right and left KD temporal states $\overrightarrow{\Upsilon}$ and $\overleftarrow{\Upsilon}$ obtained from the left and right  KD temporal tomography correspond to the left and right reduced states of $\overleftrightarrow{\Upsilon}$, respectively.

The doubled and left/right MH temporal states can be regarded as the Hermitianized versions of the corresponding KD temporal states. In the two-time case, the left/right MH temporal states coincide with the PDO, as will be discussed later.

From our definition of the right temporal state, the temporal Born rule (referring to the inner product between measurement operators and state operators, analogous to the spatial case) reproduces the right temporal KD quasiprobability distribution:
\begin{equation}
    \Tr \big[ (\Pi_{b_n} \cdots \otimes \Pi_{b_0}) \, \overrightarrow{\Upsilon} \big] 
    = \QR(b_n, \ldots, b_0).
\end{equation}
For $\overleftarrow{\Upsilon}$ and $\overleftrightarrow{\Upsilon}$, the temporal Born rule analogously gives the left and the doubled temporal KD quasiprobability distributions, respectively. 
Furthermore, since the doubled $\overleftrightarrow{\Upsilon}$ coincides with the doubled density operator, it can also generate the LvN distribution when one adopts the doubled measurement setting \cite{jia2024spatiotemporal}. 
For MH temporal state, the temporal Born rule give temporal MH  quasiprobability distributions.
This tomographic understanding of the spatiotemporal state provides a generalized framework that unifies the existing formalisms of temporal states.

All of the above definitions of temporal states can be directly extended to the spatiotemporal setting by considering the spatiotemporal quasiprobability distributions and the correlators derived from them, thereby yielding spatiotemporal states.  
Before we proceed, we summarize the properties of these spatiotemporal states as follows:

\begin{theorem}
For a given spatiotemporal quantum process,  we introduce the following spatiotemporal states:
\begin{itemize}
    \item Left, right, and doueld KD spatiotemporal states $\overleftarrow{\Upsilon}$, $\overrightarrow{\Upsilon}$ and $\overleftrightarrow{\Upsilon}$;
    \item Left/right and doubled MH spatiotemporal states $\Upsilon^{\rm MH}$ and $\overleftrightarrow{\Upsilon}^{\rm MH}$.
\end{itemize}
The have the following properties:

1. The KD spatiotemporal states satisfy the following relation
\begin{align}
    & \overrightarrow{\Upsilon} = \Tr_{L} \overleftrightarrow{\Upsilon}, \quad 
      \overleftarrow{\Upsilon} = \Tr_{R} \overleftrightarrow{\Upsilon}, \quad
      \overrightarrow{\Upsilon} = \overleftarrow{\Upsilon}^{\dagger}.
\end{align}
The fixed-time state $\rho_{t_k}$ (which is density operators) can be obtained from these spatiotemporal states by taking a partial trace:
\begin{equation}
  \rho_{t_k} = \Tr_{t_n, \ldots, \widehat{t_k}, \ldots, t_0} \overrightarrow{\Upsilon} 
               = \Tr_{t_n, \ldots, \widehat{t_k}, \ldots, t_0} \overleftarrow{\Upsilon},
\end{equation}
where $\widehat{t_k}$ indicates that the partial trace is taken over all time steps except $t_k$. Since the left and right KD spatiotemporal states are the respective reduced states of the doubled KD spatiotemporal state, the equal-time density operator can likewise be obtained from it.
Moreover, satisfy the quantum version of Kolmogorov consistency condition, for any two subsets $\mathcal{T},\mathcal{S}$ of spacetime point with non-empty overlapping $\mathcal{T}\cap \mathcal{S}\neq \emptyset$, the corresponding reduced spatiotemporal states on the overlap coincide.

2. The MH spatiotemporal states are the Hermitianized versions of the KD temporal states,
\begin{equation}
    \Upsilon^{\rm MH} = \frac{1}{2}\big(\overleftarrow{\Upsilon} + \overleftarrow{\Upsilon}^{\dagger}\big), 
    \quad 
    \overleftrightarrow{\Upsilon}^{\rm MH} = \frac{1}{2}\big(\overleftrightarrow{\Upsilon} + \overleftrightarrow{\Upsilon}^{\dagger}\big).
\end{equation}
From property~1, we see that the fixed-time state $\rho_{t_k}$ can also be obtained from the MH spatiotemporal states, which additionally satisfy the quantum analogue of the Kolmogorov consistency condition.

3. The spatiotemporal Born rule for the KD spatiotemporal state yields (for notational convenience, we focus on the temporal case):
\begin{equation}
\QR(b_n, \ldots, b_0) = \Tr \big[ (\Pi_{b_n} \otimes \cdots \otimes \Pi_{b_0}) \, \overrightarrow{\Upsilon} \big]
\end{equation}
\begin{equation}
\QL(b_n, \ldots, b_0) = \Tr \big[ (\Pi_{b_n} \otimes \cdots \otimes \Pi_{b_0}) \, \overleftarrow{\Upsilon} \big]
\end{equation}
\begin{equation}
\begin{aligned}
   & \overleftrightarrow{Q}_{\rm KD}(a_n, \ldots, a_0; b_n, \ldots, b_0) \\
    = &
\Tr \big[ (\Pi_{a_n} \otimes \cdots \otimes \Pi_{a_0}) \otimes (\Pi_{b_n} \otimes \cdots \otimes \Pi_{b_0}) \, \overleftrightarrow{\Upsilon} \big]
\end{aligned}
\end{equation}
\begin{equation}
\begin{aligned}
    &Q_{\rm LvN}(a_n, \ldots, a_0)  \\
    =
& \Tr \big[ (\Pi_{a_n} \otimes \cdots \otimes \Pi_{a_0}) \otimes (\Pi_{a_n} \otimes \cdots \otimes \Pi_{a_0}) \, \overleftrightarrow{\Upsilon} \big]
\end{aligned}
\end{equation}
For MH spatiotemporal state similar results hold
\begin{equation}
Q_{\rm MH}(b_n, \ldots, b_0) = \Tr \big[ (\Pi_{b_n} \otimes \cdots \otimes \Pi_{b_0}) \, \Upsilon^{\rm MH} \big]
\end{equation}
\end{theorem}
\begin{equation}
\begin{aligned}
   & \overleftrightarrow{Q}_{\rm MH}(a_n, \ldots, a_0; b_n, \ldots, b_0) \\
    = &
\Tr \big[ (\Pi_{a_n} \otimes \cdots \otimes \Pi_{a_0}) \otimes (\Pi_{b_n} \otimes \cdots \otimes \Pi_{b_0}) \, \overleftrightarrow{\Upsilon}^{\rm MH} \big]
\end{aligned}
\end{equation}
\begin{equation}
\begin{aligned}
    &Q_{\rm LvN}(a_n, \ldots, a_0)  \\
    =
& \Tr \big[ (\Pi_{a_n} \otimes \cdots \otimes \Pi_{a_0}) \otimes (\Pi_{a_n} \otimes \cdots \otimes \Pi_{a_0}) \, \overleftrightarrow{\Upsilon}^{MH}\big]
\end{aligned}
\end{equation}

\begin{proof}
The proof follows directly from the identity $\Tr(\sigma_{\mu}\sigma_{\nu}) = d\,\delta_{\mu\nu}$ together with Lemma~\ref{lemma:QDleftright} and Kolmogorov consistency condition.  
Since $\overleftrightarrow{\Upsilon}$ coincides with the doubled density operator, a detailed account of the spatiotemporal Born rule (for the LvN distribution) can be found in Ref.~\cite{jia2024spatiotemporal}.
\end{proof}

Since $\overleftrightarrow{\Upsilon}$ coincides with the doubled density operator, as discussed in detail in Ref.~\cite{jia2024spatiotemporal}, we shall henceforth focus primarily on $\overrightarrow{\Upsilon}$, $\overleftarrow{\Upsilon}$, $\Upsilon^{\rm MH}$, and $\overleftrightarrow{\Upsilon}^{\rm MH}$. Moreover, by Theorem~2, since $\overleftarrow{\Upsilon} = \overrightarrow{\Upsilon}^{\dagger}$, it suffices, for the KD spatiotemporal state, to consider $\overrightarrow{\Upsilon}$ alone.  
We first derive an analytic expression for $\overrightarrow{\Upsilon}$ in terms of the multi-time evolution and the initial state.

\begin{theorem}
        For operators $M \in \mathbf{B}(\mathcal{H}_B \otimes \mathcal{H}_A)$ and $N \in \mathbf{B}(\mathcal{H}_C \otimes \mathcal{H}_B)$, we introduce the operation
    \begin{equation}
        N \star M := (N_{CB} \otimes \mathds{I}_A) (\mathds{I}_C \otimes M_{BA}),
    \end{equation}
    where we use the subscript to emphasize the underlying Hilbert spaces.
    Then, for density operators $\rho_{t_0} \in \mathcal{H}_{t_0}$ and Jamio{\l}kowski operators
    $
        J[\mathcal{E}_{t_{i} \leftarrow t_{i-1}}] = \sum_{k, l} \mathcal{E}_{t_{i} \leftarrow t_{i-1}}\big(|k\rangle\langle l|\big) \otimes |l\rangle\langle k|,
    $
    the temporal state arising from the KD quasiprobability distribution $\QR$ is of the form
    \begin{equation} \label{eq:recurRep}
       \overrightarrow{ \Upsilon}_{t_k \cdots t_0} = J[\mathcal{E}_{t_n \leftarrow t_{n-1}}] \star \cdots \star J[\mathcal{E}_{t_1 \leftarrow t_0}] \star \rho_{t_0}.
    \end{equation}
    This can also be understood recursively as
    \begin{equation}
        \overrightarrow{ \Upsilon}_{t_k \cdots t_0} = J[\mathcal{E}_{t_k \leftarrow t_{k-1}}] \star \Upsilon_{t_{k-1} \cdots t_0},
    \end{equation}
    with $\Upsilon_{t_0} = \rho_{t_0}$.
\end{theorem}

\begin{proof}
We first establish the expression in the two-time case. 
To this end, we use the fact that for any operators $A_{t_0}$ and $A_{t_1}$,  
\begin{equation}\label{eq:CJmap}
    \Tr_{t_A}\!\left[ J[\mathcal{E}]_{t_B t_A}\, (J_{t_B} \otimes K_{t_A}) \right] 
= \mathcal{E}(K_{t_0})\, J_{t_1}.
\end{equation}
It then follows that
\[
\Tr_{t_1 t_0} \big( \overrightarrow{\Upsilon}_{t_1 t_0} \, (\sigma_{\mu_1} \otimes \sigma_{\mu_0}) \big) 
= \Tr_{t_1}\!\big( \mathcal{E}_{t_1 \leftarrow t_0}(\rho_{t_0}\sigma_{\mu_0}) \, \sigma_{\mu_1} \big) 
= \overrightarrow{T}^{\mu_1,\mu_0}.
\]  
This completes the proof.

For the multi-time case, we compute
\[
\Tr_{t_n \cdots t_0} \!\left[ 
\big( J[\mathcal{E}_{t_n \leftarrow t_{n-1}}] \star \cdots \star J[\mathcal{E}_{t_1 \leftarrow t_0}] \star \rho_{t_0} \big) 
\big( \sigma_{\mu_n} \otimes \cdots \otimes \sigma_{\mu_0} \big) 
\right].
\]
This is done iteratively: we first take $\Tr_{t_0}$, then $\Tr_{t_1}$, and so on.  
At each step, we make use of Eq.~\eqref{eq:CJmap}.  
Proceeding in this way, the final result is seen to be precisely
$\overrightarrow{T}^{\mu_n, \ldots, \mu_0}$.
\end{proof}

\begin{figure}[t]
    \centering
    \includegraphics[width=0.5\linewidth]{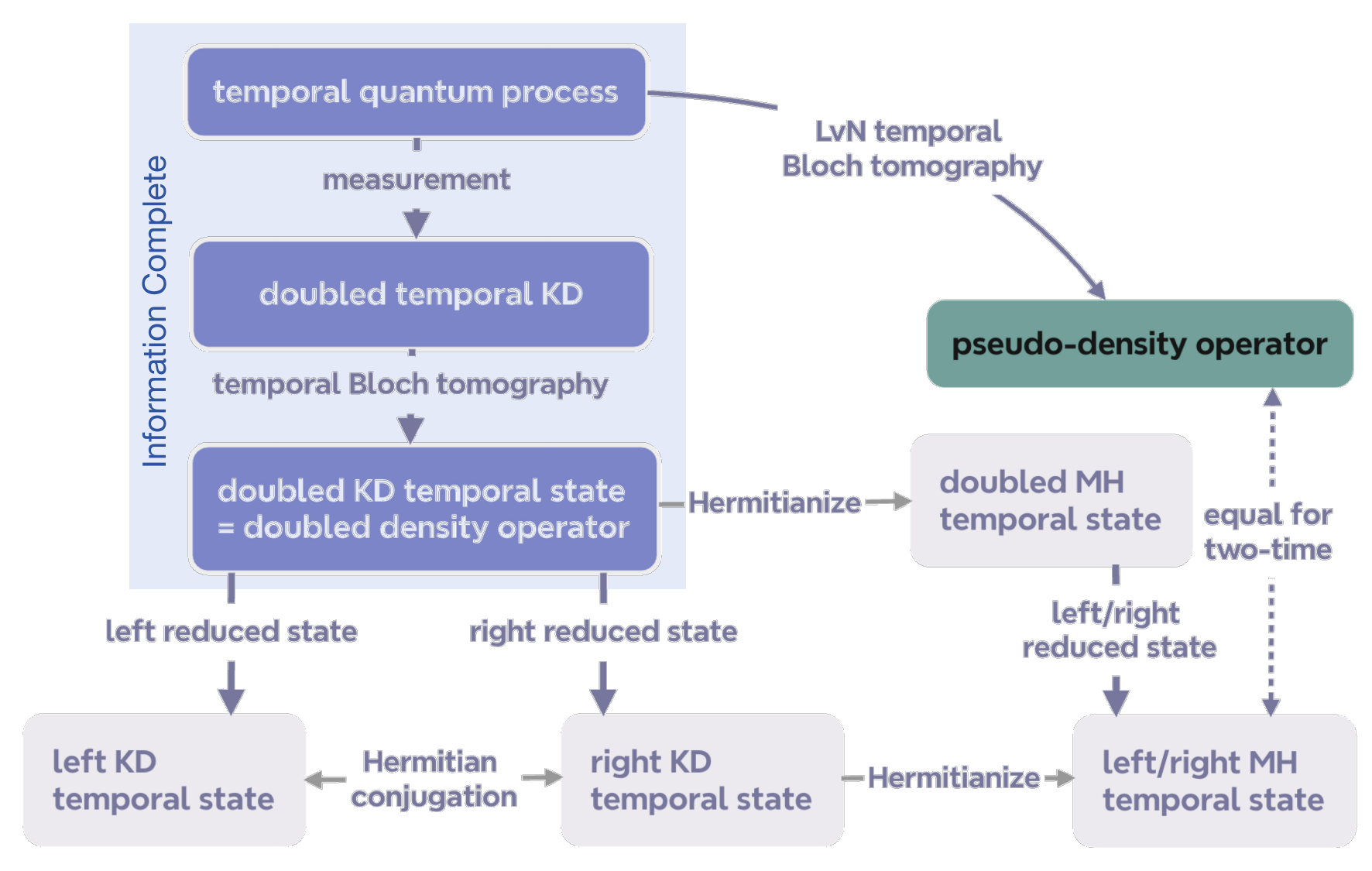}
\caption{The relationships between different temporal states arising from temporal Bloch tomography. Starting from the doubled temporal KD quasiprobability distributions, all temporal states can be obtained.}
    \label{fig:relation}
\end{figure}

The temporal state obtained from the MH quasiprobability distribution plays a crucial role in establishing a connection between the KD temporal state and the PDO formalism. In the two-time setting, the left and right MH states coincide with the PDO.

\begin{corollary}
    If we implement left/right temporal tomography from the MH quasiprobability distribution, the resulting left/right  MH temporal state takes the form
    \begin{equation}\label{eq:MHtemporalSate}
        \Upsilon_{MH} = \frac{1}{2} ( \overrightarrow{\Upsilon} + \overleftarrow{\Upsilon} )=  \frac{1}{2} (  \overrightarrow{\Upsilon} + (\overrightarrow{\Upsilon})^{\dagger}).
    \end{equation}
    The temporal Born rule for this state gives us the temporal MH quasiprobability distribution
    \begin{equation}
        Q_{\rm MH}(b_n,\cdots,b_0) = \Tr ((\Pi_{b_n} \cdots \otimes \Pi_{b_0}) \,  \Upsilon^{\rm MH}).
    \end{equation}
    This construction highlights a connection with the PDO formalism~\cite{fitzsimons2015quantum}. For the two-time PDO $R_{t_1 t_0}$, one has~\cite{zhao2018geometry,fullwood2022quantum,Parzygnat2023pdo,Liu2025PDO}
    \begin{equation}\label{eq:PDOMH}
        R_{t_1 t_0} = \frac{1}{2} \left( \overrightarrow{\Upsilon}_{t_1 t_0} + \overrightarrow{\Upsilon}_{t_1 t_0}^\dagger \right), 
    \end{equation}
    showing that the PDO is obtained by Hermitianization of the temporal state associated with the KD distribution. In the two-time case, $\Upsilon_{MH}$ coincides with the PDO.
    However, for general multi-time case, left/right MH temporal state is not the same as PDO.
\end{corollary}

\begin{proof}
Eq.~\eqref{eq:PDOMH} follows from the Jordan product representation of the PDO:
\[
R_{t_1t_0} = \frac{1}{2} \big( J[\mathcal{E}_{t_1 \leftarrow t_0}] \star \rho_{t_0} + \rho_{t_0} \star J[\mathcal{E}_{t_1 \leftarrow t_0}] \big).
\]
Since the multi-time PDO can be defined recursively as
\[
R_{t_n\cdots t_0} = \frac{1}{2} \big( J[\mathcal{E}_{t_{n} \leftarrow t_{n-1}}] \star R_{t_{n-1}\cdots t_0} + R_{t_{n-1}\cdots t_0} \star J[\mathcal{E}_{t_{n} \leftarrow t_{n-1}}] \big), \quad n>1,
\]
the multi-time PDO no longer coincides with the left or right MH temporal state.
\end{proof}

For the multi-time case, the PDO is given by
\begin{equation}
    \Upsilon^{\mathrm{LvN}}_{t_n\cdots t_0}= \frac{1}{2}\left\{ J\bigl[\mathcal{E}_{t_n\leftarrow t_{n-1}}\bigr],\Upsilon^{\mathrm{LvN}}_{t_{n-1} \cdots t_0} \right\},
\end{equation}
with $\Upsilon^{\mathrm{LvN}}_{t_0}=\rho_{t_0}$. Expanding this expression reveals that the PDO is a linear combination of products of the initial state and the Jamiołkowski operators of quantum channels at all time steps. Since these operators do not commute, the sum contains many terms differing by the order of multiplication. For example, for a three-step quantum process:
\begin{equation}
     \Upsilon^{\mathrm{LvN}}_{t_2t_1 t_1} = \frac{1}{4} \left( 
         J\bigl[\mathcal{E}_{t_2 \leftarrow t_1}\bigr] \star J\bigl[\mathcal{E}_{t_1 \leftarrow t_0}\bigr] \star \rho_{t_0} 
         + J\bigl[\mathcal{E}_{t_2 \leftarrow t_1}\bigr] \star \rho_{t_0} \star J\bigl[\mathcal{E}_{t_1 \leftarrow t_0}\bigr] 
         + J\bigl[\mathcal{E}_{t_1 \leftarrow t_0}\bigr] \star \rho_{t_0} \star J\bigl[\mathcal{E}_{t_2 \leftarrow t_1}\bigr] 
         + \rho_{t_0} \star J\bigl[\mathcal{E}_{t_1 \leftarrow t_0}\bigr] \star J\bigl[\mathcal{E}_{t_2 \leftarrow t_1}\bigr]
     \right).
\end{equation}
In contrast, left/right MH states contain only a sum of two terms with fixed product order, as given in Eq.~\eqref{eq:recurRep}. For the three-time-step case:
\begin{equation}
    \Upsilon^{\mathrm{MH}}_{t_2t_1 t_0} = \frac{1}{2} \left( 
        J\bigl[\mathcal{E}_{t_2 \leftarrow t_1}\bigr] \star J\bigl[\mathcal{E}_{t_1 \leftarrow t_0}\bigr] \star \rho_{t_0} 
        + \rho_{t_0} \star J\bigl[\mathcal{E}_{t_1 \leftarrow t_0}\bigr] \star J\bigl[\mathcal{E}_{t_2 \leftarrow t_1}\bigr]
    \right).
\end{equation}
This streamlined expression captures key physical properties of temporal quantum processes.

\bibliographystyle{apsrev4-1-title}
\bibliography{Jiabib}

%% file: mainDef.tex
\usepackage{xcolor}
\definecolor{darkblue}{RGB}{40, 85, 120} 
\definecolor{quantumviolet}{HTML}{53257F} 
\definecolor{quantumgray}{HTML}{555555} 
\definecolor{quantumgreen}{HTML}{007474} 
\definecolor{quantumblue}{HTML}{3A5FCD} 
\definecolor{quantumpurple}{HTML}{8A2BE2} 
\definecolor{warmorange}{HTML}{FF8C42} 
\definecolor{violetpurple}{HTML}{8A2BE2} 
\definecolor{teal}{HTML}{008080} 
\definecolor{softbeige}{HTML}{ECE5D7} 
\definecolor{charcoalgray}{HTML}{2E2E2E} 
\definecolor{coralred}{HTML}{FF6F61} 
\definecolor{brightyellow}{HTML}{FFD700} 
\definecolor{cyanblue}{HTML}{1C77C3} 
\definecolor{deepbluegray}{HTML}{2C3E99} 
\definecolor{softblue}{HTML}{AFCBFF} 

\usepackage{amsmath}
\usepackage{amssymb}
\usepackage{amsthm}
\usepackage{bm}
\usepackage{braket}
\usepackage{xcolor}
\usepackage{graphicx}
\usepackage[caption=false]{subfig}
\usepackage[colorlinks, linktocpage, hypertexnames = false]{hyperref}

\usepackage{dsfont}

\hypersetup{linkcolor = {gray!80!black},
citecolor = {blue!90!black}, 
urlcolor = {gray!60!black}}

\newcommand{\orcidd}[1]{\href{https://orcid.org/#1}{\includegraphics[width=8pt]{orcid}}} 

\usepackage{tikz}
\usetikzlibrary{positioning,intersections}
\usetikzlibrary{calc}
\usetikzlibrary{shapes.geometric}
\usepackage{braids}
\usepackage{tikz-cd}
\usetikzlibrary{shapes.geometric,shapes.misc}
\usetikzlibrary{arrows,matrix,calc,scopes,decorations.markings,snakes}
\usetikzlibrary{arrows.meta}
\usetikzlibrary{intersections, patterns,fit} 
\usetikzlibrary{decorations.pathreplacing,angles,quotes}

\usepackage{qcircuit}

\theoremstyle{theorem}
\newtheorem{theorem}{Theorem}
\newtheorem{lemma}{Lemma}
\newtheorem{definition}{Definition}

\theoremstyle{remark}

\newcommand{\Herm}{\mathsf{Herm}}

\newcommand\Tr{\operatorname{Tr}}

\usepackage{euscript}

\usepackage{mathptmx}

\usepackage[bbgreekl]{mathbbol} 
\DeclareMathAlphabet{\mathcal}{OMS}{cmsy}{m}{n}

\usepackage[noframe]{showframe}
\usepackage{framed}
\usepackage{lipsum}
 {\endMakeFramed}
\definecolor{shadecolor}{gray}{0.9} 